\documentclass{article}

\usepackage{arxiv}

\usepackage{comment}
\usepackage[utf8]{inputenc} 
\usepackage[T1]{fontenc}    
\usepackage{hyperref}       
\usepackage{url}            
\usepackage{booktabs}       
\usepackage{amsfonts}       
\usepackage{amsthm}
\usepackage{nicefrac}       
\usepackage{microtype}      
\usepackage{lipsum}
\usepackage{url,amsmath,setspace,amssymb,fullpage}
\usepackage{xcolor}
\usepackage{graphicx} 
\usepackage{booktabs} 
\usepackage{tikz}
\usetikzlibrary{shapes}
\usepackage{mathtools}
\usepackage{float}

\newtheorem{lemma}{Lemma}
\newtheorem{theorem}[lemma]{Theorem}
\newtheorem{corollary}[lemma]{Corollary}
\newtheorem{definition}[lemma]{Definition}
\newtheorem{proposition}{Proposition}

\numberwithin{lemma}{section}

\newcommand{\F}{\mathbb{F}}
\newcommand{\Fl}{\mathcal{F}_{\X}(\ell)}

\newcommand{\IS}{I(\mathbb{S})}
\newcommand{\HM}{\mathbb{S}^{\ell}(\F_{q})}
\newcommand{\CH}{C^{\mathbb{S}}(\ell)}
\newcommand{\CA}{C^{\mathbb{A}}(\ell,2\ell)}
\newcommand{\X}{\mathbf{X}}
\newcommand{\Y}{\mathbf{Y}}

\title{AFFINE SYMPLECTIC GRASSMANN CODES}

\author{  \textbf{Fernando Pi\~nero Gonz\'alez} \\
  Department of Mathematics \\
  University of Puerto Rico in Ponce \\
  Ponce, PR. \\ 
  \texttt{fernando.pinero1@upr.edu} \\   \textbf{Doel Rivera Laboy}  \\
  Department of Mathematics\\
  Pontifical Catholic University of Puerto Rico\\
  Ponce, PR \\
  \texttt{driveralaboy@pucpr.edu} }

\begin{document}
\maketitle

\begin{abstract}

In this manuscript, we introduce a new class of linear codes, called affine symplectic Grassmann codes, and determine their parameters, automorphism group, minimum distance codewords, dual code and other key features. These linear codes are defined from an affine part of a polar symplectic Grassmannian. They combine polar symplectic Grassmann codes and affine Grassmann codes.

\end{abstract}


\section{Introduction}
Let $q$ be a prime power and $\F_q$ denote the finite field of $q$ elements. The Grassmannian, $\mathcal{G}_{\ell, m}$, is the collection of all subspaces of dimension $\ell$ of a vector space $V$ of length $m$. Without loss of generality, we consider $V=\F_q^m.$ The Grassmannian is a well studied object with a rich algebraic, geometric and combinatorial structure. 

It is well known that the Grassmannian may be embedded into the projective space $\mathbb{P}\left(\binom{m}{\ell}, \F_q\right)$ through the Pl\"ucker embedding, by mapping each vector space $W \in \mathcal{G}_{\ell, m}$ to a projective point $P_W$. The coordinates of the point $P_W$ are the minors of an $\ell \times m$ matrix $M_W$ whose rows span $W$. The Grassmann code is defined as the linear code generated by taking a representative for each point $p_W$ as the column vectors of a matrix. In \cite{Nogin}, Nogin studied the parameters of this linear code, denoted $C(\ell, m)$. 
Later, in \cite{BGT} Beelen, Ghorpade and H\o{}holdt introduced the linear code associated to one of the affine maps of the Pl\"ucker embedding of the Grassmannian. These codes are known as affine Grassmann codes, $C^\mathbb{A}(\ell,m)$ for $m \geq 2\ell$. Let $\delta = \ell (m-\ell)$, the authors proved the parameters of $C^\mathbb{A}(\ell,m)$ are:
$$[q^{\delta}, \binom{m}{\ell},  q^{\delta- \ell^2}\prod\limits_{i=0}^{\ell-1}q^\ell - q^i]$$
In that same article, they studied the automorphisms and counted the minimum weight codewords. Their dual codes, $C^\mathbb{A}(\ell,m)^\perp$ and related codes were studied in \cite{BGT2} and \cite{AfGrass}. Moreover, in \cite{AfGrass} the minimum weight codewords of the dual code were classified and counted.

Polar Grassmannians are special subvarieties of the Grassmannian. If $B: V\times V \rightarrow \F$ is a bilinear form, a polar Grassmannian is a subvariety of the Grassmannian where the subspaces $W \in \mathcal{G}_{\ell, m} $ satisfy the relation $B(x,y) = 0 \forall \ x,y \in W$. If the form $B$ is a quadratic form, the subvariety is known as an orthogonal Grassmannian. The subspaces satisfying $B(x,y) = 0$ are the nonsingular spaces. If the form $B$ is an alternant form, the subvariety is known as a symplectic Grassmannian and the subspaces satisfying $B(x,y) = 0$ are isotropic spaces. If $m = 2\ell$, then the symplectic Grassmannian is also known as the Lagrangian Grassmannian. If the form $B: V \times V \rightarrow \F$ is a sesquilinear form, the subvariety $\{W \in \mathcal{G}_{\ell, m} \ | \ B(v,w) = 0 \ \forall v,w \in W \}$  is known as a unitary Grassmannian or a Hermitian Grassmannian.The subspaces satisfying $B(x,y) = 0$ are isotropic spaces. In particular, we remark that the the symplectic forms $B(x,y)$ may be represented as the map $xHy^T$ where $H$ is a skew symmetric matrix. 

In \cite{Cardinali_2013}, Cardinali and Giuzzi introduced polar Grassmann codes. Polar Grassmann codes are linear codes defined from a projection of a Grassmann code (or in coding theory terms, a code puncturing) onto the isotropic spaces of the corresponding polar Grassmannian. 
In the case of $\ell =2$ Cardinali and Giuzzi determined the parameters of polar Grassmann codes under a symplectic form (polar symplectic Grassmannian)\cite{Symplectic}, under a orthogonal form (polar orthogonal Grassmannian)\cite{CARDINALI20161924}\cite{Cardinali_2018} and under a Hermitian form (polar Hermitian Grassmannian)\cite{Cardinali}. Furthermore in \cite{Symplectic} for $\ell = 3 $ and $m = 6$ Cardinali and Giuzzi provide the full weight enumerator of the corresponding symplectic Grassmann code.

Much is known about both affine Grassmann codes and polar Grassmann codes. However, in this work we study the linear code associated to an affine map of the polar symplectic Grassmannian. In particular we use the symplectic form defined as $$B(x,y) := \sum\limits_{i=1}^{\ell} x_{i}y_{2\ell+1-i} - \sum\limits_{i=\ell+1}^{2\ell} x_{i}y_{2\ell+1-i}.$$ Note that $B(x,y)$ is equal to $xAy^T$ where $$ A_{i,j} := 
  \begin{cases} 
       1 & \makebox { if } j = 2\ell+1-i \makebox { and } i\leq \ell \\
       -1 & \makebox { if } j = 2\ell+1-i \makebox { and } i> \ell \\
       0 & \makebox{ if } j \neq 2\ell+1-i
   \end{cases}.$$
   In particular there is a correspondence between a symmetric matrix $M$ and a totally isotropic subspace $W$ which is the rowspace of the $\ell \times 2\ell$ matrix $M_W = [M|J]$ where $J$ is a $(0,1)$ matrix whose nonzero elements are in the antidiagonal. As in the case of Affine Grassmann codes, the minors of the matrix $M$ are in correspondence to the $\ell$ minors of $M_W$.

To determine the parameters of affine symplectic Grassmann codes, we shall use elementary algebraic techniques. We also determine some automorphisms of the linear code and determine the minimum distance codewords of their dual code. We remark that when $\ell = 1$, the code is equivalent to the first order generalized Reed-Muller
code \cite{Delsarte1970OnGR} and therefore we work on the cases where $\ell \geq 2$.

\section{Defining affine symplectic Grassmann codes}
In this manuscript we are interested mainly in square matrices over $\F_{q}$. Matrices will be denoted by upper case letters such as $A, B, C, M, N$. Generic matrices will be denoted by $\X$ or $\Y$.

\begin{definition}
    Let $M$ be a square matrix. Let $I$ be a subset of rows of $M$ and $J$ be a subset of columns of $M$, such that $\#I = \#J$. The minor $det_{I,J}(M)$ is the determinant of the submatrix of $M$ obtained from the rows $I$ and columns $J$. 
\end{definition}
We remark that if $I = \{ i\}$ and $J = \{ j \}$ the $det_{I,J}(M)$ is the $(i,j)-th$ entry $M_{i,j}$.
Example:\\
Let $M = \begin{bmatrix}
    1&0&0\\0&1&0\\0&0&1
    \end{bmatrix}$, $I = \{1,2\}$ and $J = \{2,3\}$. Then $det_{I,J}(M) = \begin{vmatrix}
    0&0\\1&0
    \end{vmatrix} = 0$ 

Moreover, to keep notation simple when working with an $\ell$ minor, we denote the row and column sets as $\ell$ digit numbers. Under this shorthand notation, the previous example, $det_{I,J}(M) = det_{12,23}(M)$.
\begin{definition}
    A matrix over $\F_{q}$ is symmetric if $M  = M^T$.
\end{definition} If a matrix $M$ is symmetric then $M_{J,I} = M_{I,J} $ for all sets $I,J$ where $\# I = \# J$.

For $\ell \geq 1$, denote by $x_{i,j}$ a variable for $1 \leq i \leq j \leq \ell$. In this way we have defined only $\binom{\ell+1}{2}$ indeterminates. We define $\X = [X_{ij}]$ as the $\ell \times \ell$ symmetric matrix of $\binom{\ell+1}{2}$ indeterminates over $\F_{q}$ as follows  $$X_{i,j} = \begin{cases}
 x_{i,j}  & i\leq j \\
 x_{j,i} &  i>j\\
\end{cases}$$

\begin{definition}
    $\Delta_t(\ell)$ is the set of all $t\times t$ minors of the symmetric generic matrix $\X$. That is:
    $$\Delta_t(\ell) := \{ det_{I,J}(\X), I,J \subseteq [\ell], \#I = \#J = t  \} $$
\end{definition}
We remark that for $i\neq j$, $(\Delta_i(\ell) \bigcap \Delta_j(\ell)) = \emptyset$.

\begin{definition}
We define $\Fl_t$ as the subspace of $\mathbb{F}_{q}$--linear combinations of elements of $\Delta_t(\ell)$.

$$\Fl_t := \{ \sum\limits_{I,J \subseteq [\ell], \# I = \# J = t} f_{I,J} det_{I,J}(\X), f_{I,J} \in \F_{q}   \} $$
\end{definition}
When considering the appropriate generic matrix for both affine Grassmann codes and the affine Hermitian Grassmann codes, all minors $\Fl$ are linearly independent. However, this is not the case for the generic symmetric matrix case as many minors are linear combinations of other minors. In fact, $det_{1,2}(\X) = X_{1,2} = X_{2,1} = det_{2,1}(\X)$. However, $\Fl_t$ is spanned by special minors known as doset minors.

\begin{definition}
Let $A = \{a_1,a_2,..., a_n\}$ and $B = \{b_1, b_2, ..., b_n \}$ be the rowset and columnset respectively of the minor $det_{A,B}(\X)$. The minor $det_{A,B}(\X)$ is a doset minor if and only if:
$$\forall i\in [n], a_i\leq b_i.$$
\end{definition}

\begin{definition}
$\Delta(\ell)$ is the set of all minors of the matrix $\X$. That is:
    $$\Delta(\ell) := \bigcup_{t = 0}^\ell \Delta_t $$
\end{definition}

Example:\\
Let $\ell = 2$\\
  $$\X = \begin{bmatrix}
    X_{1,1}&X_{1,2}\\X_{1,2}&X_{2,2}
    \end{bmatrix},$$ \\
Then $\Delta(\ell) = \{ 1, X_{1,1}, X_{1,2}, X_{2,2}, X_{1,1}X_{2,2}-X_{1,2}^{2} \} $

\begin{definition}
We define $\Fl$ as the subspace of $\mathbb{F}_{q}$--linear combinations of elements of $\Delta(\ell)$.

$$\Fl := \{ \sum\limits_{I,J \subseteq [\ell], \# I = \# J} f_{I,J} det_{I,J}(\X), f_{I,J} \in \F_{q}   \} $$
\end{definition}

\begin{definition}
    $\mathbb{S}^{\ell}(\F_{q})$ denotes the set of all $\ell \times \ell$ symmetric matrices with entries in $\F_{q}$. That is $$\mathbb{S}^{\ell}(\F_{q}) = \{H \in \mathbb{M}^{\ell \times \ell}(\mathbb{F}_{q})| M = M^T \}.$$
\end{definition}
Now we shall define the evaluation of an element $f \in \Fl$ at any symmetric matrix $P \in \mathbb{S}^{\ell}(\F_{q})$. For any $f \in \mathbb{F}_{q}[X]$ and $P \in \mathbb{S}^{\ell}(\F_{q})$  the evaluation $f(P) $ is obtained by replacing the variable $x_{i,j}$ by the element $P_{i,j}$. From now on, we shall denote $n$ by $n = q^{\frac{\ell^2+\ell}{2}} = \# \mathbb{S}^{\ell}(\F_{q}) $. For the evaluation map we fix an arbitrary enumeration $P_1, P_2, ..., P_n$ of $\mathbb{S}^{\ell}(\F_{q})$.

\begin{definition}
    The evaluation map of $\mathbb{F}_{q}[\X]$ on $\mathbb{S}^{\ell}(\F_{q})$ is the map
    \begin{center}
        $Ev_{\mathbb{S}^{\ell}(\F_{q})}$ : $\mathbb{F}_{q}[\X] \rightarrow \F_{q}^{n}$ defined by $Ev_{\mathbb{S}^{\ell}(\F_{q})}(f) := (f(P_1),...,f(P_{n}))$.
    \end{center}
\end{definition}

We are now ready to define affine symplectic Grassmann codes.

\begin{definition}
    The affine symplectic Grassmann code $C^{\mathbb{H}}(\ell)$ is the image of $\Fl$ under the evaluation map $Ev$. That is
    
    $$ \CH := \{ Ev_{\mathbb{S}^{\ell}(\F_{q})} | f \in \Fl  \} $$
\end{definition}

\section{Calculating $\dim \CH$}
In this section we calculate the dimension of the affine symplectic Grassmann codes. First we begin by calculating the dimension of the space $\Fl$.

\begin{definition}
The Narayana numbers $N(n,k)$ are defined as follows:
$$N(n, k) = \frac{1}{n}\binom{n}{k}\binom{n}{k-1}$$
\end{definition}
\begin{definition}
The Catalan numbers $C(n)$ are defined as follows:
$$C(n) = \frac{1}{n}\binom{2n}{n} = \sum\limits_{i = 1} N(n,i)$$
\end{definition}

\begin{lemma}\cite[Lemma 2.5]{Shafiei2013ApolarityFD}
Let $\X$ be a generic $\ell \times \ell$ symmetric matrix. The dimension of $\Fl_t$ is equal to the number of $t \times t$ doset minors contained in an $\ell \times \ell$ matrix. That number is equal to the number of semi-standard fillings of a Young tableau of shape
$(t, t)$ with the numbers $\{1,2, ..., \ell \}$ which is also known as the Narayana number. Then 

$$\dim(\Fl_t) = N(\ell+1, t+1)$$
\end{lemma}

This implies that $dim(\Fl_t) = \binom{\ell+1}{t+1}\binom{\ell+1}{t}/(\ell+1)$. We extend the definition of $\Fl_t$ as follows.

\begin{proposition}\cite{Shafiei2013ApolarityFD}
$dim(\Fl) = C(\ell+1)$. Furthermore the doset minors of $\X$ are a basis of $\Fl$. 
\end{proposition}
 \begin{proof}
Note that $\Fl$ is the direct sum of $\Fl_t$.
 $$dim(\Fl) =  \sum_{t = 0}^{\ell} dim(\Fl_t) = \sum_{t = 0}^{\ell} N(\ell+1,t+1) = C(\ell+1)$$
\end{proof}
 
For example, when $\ell = 2$ the vector space $\Fl$ is spanned by $\langle 1, X_{1,1}, X_{1,2}, X_{2,2}, X_{1,1}X_{2,2}-X_{1,2}^2\rangle$ and its elements are multivariate polynomials of the form
$$f_{\emptyset, \emptyset}+f_{1,1}X_{1,1}+ f_{1,2}X_{1,2}+ f_{2,2}X_{2,2}+ f_{12,12}(X_{1,1}X_{2,2}-X_{1,2}^{2})$$ where $f_{\emptyset, \emptyset},f_{1,1}, f_{1,2},f_{2,2}, f_{12,12}, \in \F_q$.




Recall that if $\X$ is any matrix, and the sets $I = \{i_1 < i_2 < \cdots < i_s\}$, and $J = \{j_1 < j_2 < \cdots < j_s\}$ are sets of rows and columns of $\X$ respectively, then $$det_{I,J}(\X) = \sum\limits_{\sigma \in S_{s}} (-1)^{sgn(\sigma)}\prod_{a = 1}^s X_{i_{\sigma(a)}, j_{\sigma(a)}}.$$ The terms of the minors are indexed by permutations $\sigma \in S_s$. The monomial $\prod_{a = 1}^s X_{i_{\sigma(a)}}$ is the monomial corresponding to $\sigma$. We shall denote the monomial corresponding to $\sigma $ by $M_{\sigma}$.


\begin{definition} 
Denote by $\IS$ the ideal generated by $X_{i,j}^q - X_{i,j}$ where $i \leq j$ and $X_{j,i}-X_{i,j}$ where $j < i$. That is $$\IS := \langle X_{i,j}^q - X_{i,j}, X_{i,j} - X_{j,i} \rangle.$$ 
\end{definition}
The generic symmetric matrix satisfies the equations $X_{j,i} = X_{i,j}$ but not the field equations $X_{i,j}^q-X_{i,j}$. The set of symmetric matrices $\mathbb{S}^{\ell}(\F_{q})$ is precisely the set of solutions to all polynomial equations generating $\IS$. In fact $\IS$ is the kernel of the evaluation map on $\mathbb{S}^{\ell}(\F_{q})$. Our aim is to prove that no nonzero function in $\Fl$ is in $\IS$. We use Gr\"obner bases and multivariate polynomial division.

\begin{definition}

We order the variables $X_{i,j}$ as follows: $$X_{i,j} \prec_0 X_{i',j'} \makebox { if and only if } i< i' \makebox{ or } i = i' \makebox{ and } j < j'.$$
\end{definition}

The following remark follows easily from the fact that all leading terms of the polynomials are pairwise distinct.

\begin{definition}
Let $X$ be a generic $\ell \times \ell$ matrix. Let $X_{i,j}$ denote its entries. We define the monomial order $\prec_{lex}$ as the  lexicographical order where the variables are ordered according to $\prec_0$.
\end{definition}
The order $\prec_{lex}$ implies that monomials are ordered first according to the degree of $X_{\ell, \ell}$ then the degree of $X_{\ell, \ell, -1}$ until we compare all degrees on each variable.  


\begin{lemma}
The polynomials $X_{i,j}^q - X_{i,j}$ where $i \leq j$ and $X_{j,i}-X_{i,j}$ where $j < i$ are a Gr\"obner basis for $\IS$ under $\prec_{lex}$. 

\end{lemma}
\begin{proof} Note that $lt_{\prec_{lex}}(X_{i,j}^q - X_{i,j}) = X_{i,j}^q $ for $i \leq j$ and $lt_{\prec_{lex}}(X_{j,i}-X_{i,j}) = X_{j,i}$ for $j > i$. All leading terms of the ideal $\IS$ have monomials with disjoint variables. \cite[Theorem 2.9.4]{CLO2007} implies that the set is a Gr\"obner basis for $\IS$. \end{proof}

\begin{lemma}
Let $I = \{i_1 < i_2 < \cdots < i_s\}$ and  $J = \{j_1 < j_2 < \cdots < j_s\}$, where $s \leq \ell$ , $I,J\subseteq [\ell]$ and $i_a \leq j_a$. The leading term of $det_{I,J}(\X)$ under $\prec_{lex}$ is the monomial corresponding to the identity $\prod_{a = 1}^s X_{i_{a}, j_{a}}$.  
\end{lemma}

\begin{proof}
Recall that $M_{id} =\prod_{a = 1}^s X_{i_{a}, j_{a}}$ denotes the term corresponding to the identity permutation and $M_\sigma = \prod_{a = 1}^s X_{i_{\sigma(a)}, j_{\sigma(a)}} $ denotes the term corresponding to $\sigma \in S_s$. We determine that $M_\sigma \prec_{lex} M_{id}$. Ties are broken according by checking the degree of the variable with the highest row index, and then according to the highest column index. Suppose that $\sigma$ is not the identity permutation. This implies there exists $r$ such that $\sigma(u) =u$ for $r\leq u \leq s$ and $ \sigma(r-1) \neq r-1$. Note that because $\sigma(u) =u$ for all $u> r-1$ it follows that $\sigma(r-1) < r-1$. When comparing the variables of $M_{id}$ with the ones in $M_{\sigma}$, both contain $X_{i_s,j_s}, X_{i_{s-1}, j_{s-1}}, \ldots, X_{i_r, j_r}$. But $X_{i_{r-1}, j_{r-1}}$ appears in $M_{id}$ whereas $X_{i_{r-1},j_{\sigma(r-1)}}$ where $\sigma(r-1) < r-1$ appears in $M_\sigma$. Therefore $M_\sigma \prec_{lex} M_{id}$. \end{proof}

In the next lemma we relate the leading term of $det_{I,J}$ with the leading term of its remainder.

\begin{lemma}
Let $det_{I,J}(\X)$ be a doset minor. Then $$lt_{\prec_{lex}}(det_{I,J}(\X)) = lt_{\prec_{lex}}(det_{I,J}(\X) \mod \IS)  $$

\end{lemma}
\begin{proof}
Let $I = \{i_1 < i_2 < \cdots < i_s\}$ and  $J = \{j_1 < j_2 < \cdots < j_s\}$, where $s \leq \ell$. Since $det_{I,J}(\X)$ is a doset minor, this implies that $i_a \leq j_a$ for any $1 \leq a \leq s$. Computing $f \mod \IS$ is done by replacing any variables $X_{j,i}$ where $j > i$ by the transposed variable $X_{i,j}$ and then reducing modulo the field equations $X_{i,j}^q-X_{i,j}$ where $i \leq j$. Because $det_{I,J}$ is a doset minor, no variable of the leading term $lt_{\prec_lex}(det_{I,J})(\X) = \prod_{a = 1}^s X_{i_{a}, j_{a}}$ is substituted by their transposed variables. Since the degree on each variable is $1$, there are no changes to the term $M_{id}$ when dividing by $X_{i,j}^q-X_{i,j}$. Since replacing the variable $X_{j,i}$ by $X_{i,j}$ will make the ordering of the variable decrease, the leading term is unchanged. Therefore $$lt_{\prec_{lex}}(det_{I,J}(\X)) = lt_{\prec_{lex}}(det_{I,J}(\X) \mod \IS) .$$

\end{proof}


\begin{lemma}\label{lem:ideal}
There is no linear combination of doset minors that can be written as an element in $\IS.$
 
 
\end{lemma}
\begin{proof}

 Let $f \in \Fl.$ Consider $f_S = f \mod \IS$, the reduction of $f$ modulo the ideal $\IS$. Since all doset minors have different leading terms under $\prec_{lex}$ and these leading terms are all distinct, it follows that if $f \neq 0$, then $f$ it has a nonzero leading term. Since $f \mod \IS$ has the same nonzero leading term as $f$ it means that $f \mod \IS$ has to be different from zero. This implies $f \not \in \IS$.
\end{proof}

\begin{lemma}\label{lem:inyective}
 The evaluation map $Ev_{\mathbb{S}^{\ell}(\F_{q})}: \Fl \rightarrow \F_q^n$ is injective.
\end{lemma}
\begin{proof}
Note that the ideal $\IS$ is precisely the ideal of polynomial functions which vanish on $\HM$. That is $ev(f) = 0$ if and only if $f \in \IS.$ As Lemma \ref{lem:ideal} implies there is no nonzero element in both $\Fl$ and $\IS$, we then have that $Ker(ev) = \Fl \cap \IS = \{0\}.$ Thus the evaluation map is injective. 

\end{proof}

From our discussion of the injectivity of the map $Ev$ we have the following corollary.

\begin{corollary}
The dimension of the affine symplectic Grassmann code is $dim(\CH)  = C(\ell).$

\end{corollary}\begin{proof} From the definition of $\CH$ it follows that $\CH = Im(Ev_{\mathbb{S}^{\ell}(\F_{q})})$. Since $\dim Ker(Ev_{\mathbb{S}^{\ell}(\F_{q})}) = 0$ and $\dim \Fl = C(\ell)$ the corollary follows. \end{proof}

\section{Automorphisms of $\CH$}
In this section, we determine some automorphisms of $\Fl$ and $\CH$. Our aim is to use the automorphisms to impose certain conditions on $f \in \Fl$ without loss of generality.
\begin{definition}

    The affine Grassmann code $\CA$ is obtained by evaluating all linear combinations of minors onto all $\ell \times \ell$ matrices over $\F_q$.
    

\end{definition}

\begin{definition}
Let $C$ be a code of length $n$. We say that a permutation $\sigma \in S_n$ is an automorphism of $C$ if and only if $$(c_1, c_2, \ldots, c_n) \in C \makebox{ if and only if } (c_{\sigma(1)}, c_{\sigma(2)}, \ldots, c_{\sigma(n)}) \in C. $$

The group of such automorphisms is determined by $Aut(C).$
\end{definition}

We recall the automorphism group of $\CA$.

\begin{proposition}\cite[Lemma 7]{BGT2}

The automorphism group $Aut(\CA)$ contains the following permutations:
\begin{itemize}
    \item For $A \in GL_{\ell}(\F_q), \X \mapsto A\X$.
    \item For $B \in GL_{\ell}(\F_q), \X \mapsto \X B$.
    \item For $M \in \mathbb{M}^{\ell \times \ell}(\F_q), \X \mapsto \X + M$.
    \item $\X \mapsto \X^T$.
\end{itemize}

\end{proposition}

As a consequence of the previous lemma, we have the following:

\begin{lemma}

The automorphism group $Aut(\CH)$ contains the group generated by the following permutations
\begin{itemize}
    \item For $A \in GL_{\ell}(\F_{q}), \X \mapsto A^{T}\X A$.
    \item For $M \in \HM, \X \mapsto \X + M$.
\end{itemize}

\end{lemma}

\begin{proof}
Note that $\CH$ is obtained from $\CA$ by removing all matrices which are not symmetric. That is the code $\CH$ is a puncturing of the code $\CA$ at the matrices in $\mathbb{M}^{\ell \times \ell}(\F_{q}) \setminus \HM$ to obtain $\CH$. Therefore the permutations in $Aut(\CA)$ fixing the subset $\HM$ are permutations of $Aut(\CH)$.
\end{proof}


One of the most important parameters of a linear code is its minimum distance, which is defined as follows:

\begin{definition}
The \emph{Hamming distance} of the vectors $x = (X_{1,1}, X_{1,2}, \ldots, x_n)$ and $y = (Y_{1,3}, Y_{2,3}, \ldots, y_n)$ is the number of positions in which $x$ and $y$ differ. That is:

$$d(x,y) := \#\{ i \ | \ x_i \neq y_i  \} $$

\end{definition}
For example: If $x = (1101)$ and $y = (1000)$, then $d(x,y) = 2$ because they differ in the second and fourth position.

\begin{definition}
The \emph{weight} of $x = (X_{1,1}, X_{1,2}, \ldots, x_n)$ is the number of positions in which $x_i \neq 0$. That is:

$$w(x) := \#\{ i \ | \ x_i \neq 0  \} $$

\end{definition}

Note that the weight of a vector is the same as its distance to the zero vector. That is $w(x) = d(x,0)$. We now state the definition of the minimum distance of a code.

\begin{definition}

Let $C$ be a code. Then the \emph{minimum distance} of $C$ is the minimum number of positions in which any two distinct elements of $C$ differ.

$$d(C) = \min\limits_{x,y \in C} d(x,y). $$

\end{definition}




We remark that for a linear code, the minimum distance is equivalent to the smallest weight of a nonzero codeword. This particular definition will be the one used throughout the following sections and we set the following notation.

\begin{definition}
Let $f \in \Fl$ We define the weight of $f$ as $$wt(f) = \# \{ S \in \HM \ | \ f(S) \neq 0 \}. $$
\end{definition}

In the next two sections of the paper we work out a proof by induction for the minimum distance. We use polynomial evaluation and bounds from the fundamental theorem of Algebra to determine $d(\CH).$

\section{The action of $Aut(\CH)$ on $\Fl$}

To compute $d(\CH)$  we shall treat $f(\X) \in \Fl$ as a multivariate polynomial. We estimate the number of zeroes of $f(\X)$ using partial evaluations and bounds on the number of zeroes of certain linear and quadratic polynomials. It is much easier to first use the automorphism group of $\CH$ to find conditions on $supp(f)$ that hold without loss of generality.  

\begin{lemma}\label{lem:hyperboliczeroes}
 Let $a,b,\lambda \in \F_q$, where $\lambda\neq 0$. Then
 \begin{itemize}
     \item The equation $(T_1+a)(T_2+b) = 0$ has $2q-1$ solutions over $\F_q$.
     \item The equation $(T_1+a)(T_2+b) = \lambda $ has $q-1$ solutions over $\F_q$.
 \end{itemize}
\end{lemma}
\begin{proof}
We begin with the case $(T_1+a)(T_2+b) = 0$. In this case then either $(T_1+a)= 0$ or $(T_2+b) = 0$. If $T_1 = -a$, any of the $q$ values for $T_2$ is a solution to the equation. Similarly, for $T_2 = -b$, any of the $q$ values for $T_1$ is a solution to the equation. The solution $T_1 = -a, T_2 = -b$ is counted twice this implies we have $q+q-1 = 2q-1$ total solutions to the equation.

Now we consider $(T_1+a)(T_2+b) = \lambda \neq 0$. If $T_1 = -a$ then the equation becomes $(a-a)(T_2+b) = \lambda$ which has no solution.  If $T_1 = \alpha$  where $\alpha$ is any element of $\F_q$ except $\alpha = -a$, then for there is exactly one value of   $T_2$ (namely $T_2 = \frac{\lambda}{\alpha+a}-b$) such that the equation is satisfied. 

As for any of $q-1$ values $\alpha \neq -a$ for $T_1$, we find exactly one value of $T_2$ such that $(T_1+a)(T_2+b) = \lambda $ is satisfied, it is established that $(T_1+a)(T_2+b) = \lambda $ has $q-1$ solutions. \end{proof}

We shall also need the following lemma on the number of solutions to a particular system of polynomial equations over $\F_{q}$.

\begin{lemma}\label{lem:systemsols}
Let $a_1, a_2, \ldots, a_n \in \F_q$ and for $1 \leq i,j \leq n$ let $b_{i,j} \in \F_{q}$. Then the system of polynomial equations given by $$T_i^{2} = a_i, 1 \leq i \leq n $$
$$T_iT_j = b_{i,j} $$ has at most $2$ solutions.
\end{lemma}
\begin{proof}

If all $a_i$'s satisfy $a_i =  0,$ then the only possible solution is $T_i = 0.$ If there is some $a_s \neq 0$, then there are at most $2$ values for $T_s$ which satisfy $T_s^{2} = a_s$ In this case, for each solution to $T_s^{2} = a_s$, there is at most one value (namely $T_s = c_s$, $T_i = \frac{b_{i,s}}{c_s^q}$ which satisfies the equation $T_i T_s = b_{i,s}$. Therefore there are at most $2$ solutions to the system of equations. \end{proof}

Now we define the support of a combination of minors $f \in \Fl$ and the concept of a maximal "term". This concept will be akin to the degree in order to make the induction proof for the minimum distance.

\begin{definition}
Let $f \in \Fl$, where $$f = \sum\limits_{I,J \subseteq [\ell], \# I = \# J} f_{I,J} det_{I,J}(X).$$ The support of $f$ is defined as $$supp(f) :=\{ det_{I,J}(X) \ | \ f_{I,J} \neq 0  \}. $$

\end{definition}
Example: Let $\X = \begin{bmatrix}
    X_{1,1}&X_{1,2}\\X_{1,2}&X_{2,2}
\end{bmatrix}$. If $f = 1+X_{1,1}X_{2,2}-X_{1,2}^2$, then $supp(f) = \{ det_{\emptyset, \emptyset}(\X), det_{12,12}(\X)\}$
\begin{definition}
Let $f \in \Fl$, where $$f = \sum\limits_{I,J \subseteq [\ell], \# I = \# J} f_{I,J} det_{I,J}(\X).$$ We say a minor $det_{I,J}(\X) \in supp(f)$ is \emph{maximal} if and only if for any other minor $det_{I',J'}(\X) \in supp(f)$ we have that $I \not\subseteq I'$ or $J \not\subseteq J'$. That is the columns and rows of $det_{I,J}(\X)$ are not contained in the rows and columns of any other determinant in $supp(f)$.

\end{definition}

\begin{definition}
Let $I, J \subseteq [\ell]$. We define the \emph{spread} of the minor $det_{I,J}(\X)$ as the set $I \cup J$. 

\end{definition}

 Consider the matrix $\X = \begin{bmatrix}
    X_{1,1}&X_{1,2}& Y_{1,3}\\X_{1,2}&X_{2,2} & Y_{2,3}\\ Y_{1,3}&Y_{2,3}&Y_{3,3}
    \end{bmatrix}$ and the minor given by rows $\{1,2\}$ and columns $\{2,3\}$. That is the minor $f = det_{12,23}(\X) =  \begin{vmatrix}
    X_{1,2}& Y_{1,3}\\X_{2,2} & Y_{2,3}
    \end{vmatrix}$. The spread of the minor is  $I \cup J = \{1,2\} \cup \{2,3\} =  \{1,2,3\}$. The following lemma will prove that in several cases we can view the $\ell = 3$ case as several $\ell = 2$ cases.

In order to bound $wt(f)$ in an orderly manner, we determine conditions on $supp(f)$ which imply $wt(f)$ is too large. We begin with the following definition.

\begin{definition} We define $E_{a,b} = (e_{i,j})_{1 \leq i,j \leq \ell}$ as the following $(0,1)$--matrix where $e_{i,j}$ satisfies
$$ e_{i,j} := 
  \begin{cases} 
       1 & \makebox { if } i = a \makebox { and } b = j \\
       0 & \makebox{ otherwise } 
   \end{cases}.$$
\end{definition}

\begin{lemma}\label{lem:principalsubdeterminant}

Let $\X$ be a generic $\ell \times \ell$ symmetric matrix, let $a \in I$ and let $\gamma_{a,a} \in \F_q$.  Then $$det_{I,I}(\X + \gamma_{a,a}E_{a,a}) = det_{I ,I }(\X)  + \gamma_{a,a}det_{I \setminus \{a\}, I \setminus \{a\} }(\X)  $$
\end{lemma}
\begin{proof} The statement follows from the expansion of the determinant of $\X$ along the $a$--th column.\end{proof}

\begin{lemma}\label{lem:nonprincipalsubdeterminant}

Let $\X$ be a generic $\ell \times \ell$ symmetric matrix, let $a,b \in I$ where $a<b$ and let $\gamma_{a,b} \in \F_q$.  Then $$det_{I,I}(\X + \gamma_{a,b}E_{a,b} + \gamma_{a,b}E_{b,a}) = det_{I,I}(\X ) +(-1)^{b+a}2\gamma_{a,b}det_{I\setminus \{b\},I\setminus \{a\}}(\X) +\gamma_{a,b}^2det_{I\setminus \{a,b\},I\setminus \{a,b\}}(\X) .$$
\end{lemma}
\begin{proof} Expanding $det_{I,I}(\X + \gamma_{a,b}E_{a,b} + \gamma_{a,b}E_{b,a})$ along the cofactors in column $b$ we obtain $$det_{I,I}(\X + \gamma_{a,b}E_{a,b} + \gamma_{a,b}E_{b,a}) = det_{I,I}(\X + \gamma_{a,b}E_{b,a}) +(-1)^{a+b}\gamma_{a,b}det_{I\setminus \{a\},I\setminus \{b\}}(\X + \gamma_{a,b}E_{b,a}).$$

Now we expand $det_{I,I}(\X + \gamma_{a,b}E_{b,a})$ and $det_{I\setminus \{a\},I\setminus \{b\}}(\X + \gamma_{a,b}E_{b,a})$ along the row given by $b$ and obtain

$$det_{I,I}(\X +  \gamma_{a,b}E_{b,a}) = det_{I,I}(\X ) +(-1)^{b+a}\gamma_{a,b}det_{I\setminus \{b\},I\setminus \{a\}}(\X)$$ and
$$det_{I \setminus \{a\},I \setminus \{b\}}(\X + \gamma_{a,b}E_{b,a}) = det_{I \setminus \{a\} ,I\setminus \{b\} }(\X) +(-1)^{a+b}\gamma_{a,b}det_{I\setminus \{a,b\},I\setminus \{a,b\}}(\X).$$

Since the matrix $\X$ is symmetric, we know that $$det_{I \setminus \{a\} ,I\setminus \{b\} }(\X)  =d et_{I \setminus \{b\} ,I\setminus \{a\} }(\X)$$
Putting all the equations together we obtain:

$$det_{I,I}(\X + \gamma_{a,b}E_{a,b} + \gamma_{a,b}E_{b,a}) = det_{I,I}(\X ) +(-1)^{b+a}2\gamma_{a,b}det_{I\setminus \{b\},I\setminus \{a\}}(\X) +\gamma_{a,b}^2det_{I\setminus \{a,b\},I\setminus \{a,b\}}(\X) .$$

\end{proof}

Now we prove that when adding certain matrices the determinant function is unchanged.

\begin{lemma}\label{lem:principalsubdeterminant-nonsupport}

Let $\X$ be a generic $\ell \times \ell$ symmetric matrix, let $a \not\in I$ and let $\gamma_{a,a} \in \F_q$.  Then $$det_{I,I}(\X + \gamma_{a,a}E_{a,a}) = det_{I ,I }(\X)$$
\end{lemma}
\begin{proof} The statement follows from the fact that the matrices $\X$ and $\X+\gamma_{a,a}+E_{a,a}$ have the same entries in the set of rows and columns defined by $I$.\end{proof}

\begin{lemma}\label{lem:nonprincipalsubdeterminant-nonsupport}

Let $\X$ be a generic $\ell \times \ell$ symmetric matrix, let $a,b \not\in I$ where $a<b$ and let $\gamma_{a,b} \in \F_q$.  Then $$det_{I,I}(\X + \gamma_{a,b}E_{a,b} + \gamma_{a,b}E_{b,a}) = det_{I,I}(\X ).$$
\end{lemma}

We use Lemma \ref{lem:principalsubdeterminant}, Lemma \ref{lem:nonprincipalsubdeterminant}, Lemma \ref{lem:principalsubdeterminant-nonsupport} and Lemma \ref{lem:nonprincipalsubdeterminant-nonsupport} to impose certain conditions of $f \in \Fl$ without losing any generality on $wt(f)$. We shall use the fact that a symmetric matrix $S \in \HM$  permutes $\HM$ via the map $M \mapsto M + S.$ The matrix $S \in \HM$ also induces a permutation of $\Fl$ via the map $f(\X) \mapsto f(\X + S).$ The next lemma implies that under a certain translation, we may consider $f$ has no "terms" of second highest degree.

\begin{lemma}\label{lem:minorclear}
Let $f \in \Fl$ and $q$ be odd. Suppose $det_{I,I}$ is a maximal minor of $f$ and that $f_{I,I} = 1$. Let $S$ as the following symmetric matrix $$ s_{a,b} := 
  \begin{cases} 
       -f_{I\setminus \{a\}, I \setminus \{a\}} & \makebox { if } a = b  \\
              -(-1)^{a+b}\frac{1}{2}f_{I\setminus \{a\}, I \setminus \{b\}} & \makebox { if } a < b  \\
                     -(-1)^{a+b}\frac{1}{2}f_{I\setminus \{b\}, I \setminus \{a\}} & \makebox { if } a > b  
   \end{cases}.$$ 
   
   Then $g(\X) = f(\X + S_f)$  is an element of $\Fl$ with no $(\# I -1) \times (\# I -1)$ minors in its support whose rows and columns are contained in $I$.

\end{lemma}

\begin{proof}
Let $f, S$ be as in the statement of the theorem. Consider $g(\X)  = f(\X+S)$. Since $det_{I,I}(\X)$ is a maximal minor of $f$ the only terms which may contain $det_{I \setminus\{a\}, I \setminus \{b\} }(\X)$ are the ones coming from $det_{I,I}(\X)$ and $det_{I \setminus\{a\}, I \setminus \{b\} }(\X)$. From Lemma \ref{lem:principalsubdeterminant} and Lemma \ref{lem:nonprincipalsubdeterminant} it follows that taking $s_{a,a} = -f_{a,a}$ and $s_{a,b} = s_{b,a} = (-1)^{a+b}\frac{1}{2}f_{a,b}$ will cause the $det_{I \setminus \{a\},  I \setminus \{b\}}$ terms of $f(\X+ S)$ cancel out. Thus $g(\X) = f(\X+S)$ has no terms in its support of size $(\# I-1) \times (\# I-1)$ contained in $I \times I$.\end{proof}




Now we establish certain relations between $\CH$ and $C^\mathbb{S}(\ell+1)$. These relations determine the minimum distance in the general case. First, we recall the following notation:
\begin{definition}
We denote the elementary matrix for row addition operations $L_{i,j}(m)$ as the corresponding matrix obtained by adding m times row $j$ to row $i$.
\end{definition}
Example:\\
Let $\ell = 3$, $L_{1,2}(\alpha) = \begin{bmatrix}
    1&\alpha & 0\\0& 1 & 0\\ 0&0&1
    \end{bmatrix}$\\

In particular, the weight of a function $Ev_{\HM}(f)$ for  $f\in \Fl$ depends on three aspects: the ambient matrix space $\HM$, the size of the maximal minor on $supp(f)$ and the spread of the maximal minor on $supp(f)$. We now introduce the following notation:
\begin{definition} Denote by 
$w_{\ell,k,s}$ the minimum weight of a function $Ev_{\HM}(f)$ such that the matrices are $n \times n$, with a maximal minor of size $k \times k$ and minimal spread $s$.
\end{definition}
 As in the case of affine Hermitian Grassmann codes $C^\mathbb{H}(\ell)$, we will find the minimum distance of the code $\CH$ by inducting on the different terms $w_{\ell,k,s}$. Now we generalize some behavior of the code using the new notation.

\begin{lemma}\label{lem: reduce to spread}
Let $f \in \Fl$. Suppose that $f$ has a maximal minor of size $k$ whose spread has size $=s$. Then $wt(f) \geq q^\frac{\ell^2+\ell-s^2-s}{2}w_{s,k,s}$. 
\end{lemma}
\begin{proof}
    Let $f$ be as in the statement of the lemma. This implies there are $\ell-s$ rows and columns which do not appear in the spread of the maximal minor. Then for any of the $q^\frac{\ell^2+\ell-s^2-s}{2}$ values one can put on these columns, $f$ specializes to a combination of $s\times s$ determinants with the same maximal minors. Other nonmaximal minors may change due to the specialization. As each specialization has weight bounded by $w_{s,k,s}$ and there are $q^\frac{\ell^2+\ell-s^2-s}{2}$ specializations, the Lemma follows. 
\end{proof}

Note the following behavior in the $3\times 3$ case:\\
Let $\X = \begin{bmatrix}
    X_{1,1}&X_{1,2}& Y_{1,3}\\X_{1,2}&X_{2,2} & Y_{2,3}\\ Y_{1,3}&Y_{2,3}&Y_{3,3}
    \end{bmatrix}$ and $f = det_{12,23}(\X) =  \begin{vmatrix}
    X_{1,2}& Y_{1,3}\\X_{2,2} & Y_{2,3}
    \end{vmatrix}$\\
Note that $f \in \Fl$ is a minor of size $2$ and spread $3$. Also recall that the following map $\X \mapsto A^{T}\X A$ is an automorphism which preserves the weight of of a codeword, that is $wt(f(\X)) = wt(f(A^{T} \X A))$. \\
Then let $A = \begin{bmatrix}
    1&0&1\\0&1&0\\ 0&0&1
    \end{bmatrix}$\\
    $\begin{bmatrix}
    1&0&0\\0&1&0\\ 1&0&1
    \end{bmatrix} \begin{bmatrix}
    X_{1,1}&X_{1,2}& Y_{1,3}\\X_{1,2}&X_{2,2} & Y_{2,3}\\ Y_{1,3}&Y_{2,3}&Y_{3,3}
    \end{bmatrix}  \begin{bmatrix}
    1&0&1\\0&1&0\\ 0&0&1
    \end{bmatrix} = \begin{bmatrix}
    X_{1,1}&X_{1,2}& X_{1,1}+Y_{1,3}\\X_{1,2} &X_{2,2} & X_{1,2}+Y_{2,3}\\ X_{1,1}+Y_{1,3}& X_{1,2}+Y_{2,3}&X_{1,1}+Y_{1,3}+Y_{1,3}+Y_{3,3}
    \end{bmatrix}$\\
Then $f =  \begin{vmatrix}
    X_{1,2}& X_{1,1}+Y_{1,3}\\X_{2,2} & X_{1,2}+Y_{2,3}
    \end{vmatrix} = \begin{vmatrix}
    X_{1,2}& Y_{1,3}\\X_{2,2} & Y_{2,3}
    \end{vmatrix}+\begin{vmatrix}
    X_{1,2}& X_{1,1}\\X_{2,2} & X_{1,2}
    \end{vmatrix}=det_{12,23}(\X)-det_{12,12}(\X)$\\
Note $f(A^{T} \X A)$ is combination of minors in $\Fl$ containing a maximal minor of size $2$ and spread $3$. We may transform a function with a given maximal minor of size $m$ and spread $s$ to another function of the same weight with the maximal minor size $m$ and spread $s-1$. 
The matrix $A$  we used in the transformation may be obtained from elementary matrices. In this example we used $L_{1,3}(1)$. We split the proof in two cases: when $s = k+1$ and $s > k+1$. \begin{lemma}\label{lem:spread reduction 1}
Let $f \in \Fl$. Suppose that $f$ has a maximal minor of minimal spread, $\mathcal{M}$ of size $k$ and spread $=s=k+1$. Then there exists $g \in \Fl$ of the same weight with a maximal minor with size $k$ and spread size $k$.
\end{lemma}

\begin{proof}
Let $\mathcal{M} = det_{I,J}(\X)$ be a maximal minor of size $k$ and spread size $s>k$ in the support of $f$.
Applying a suitable permutation of rows and columns, Without loss of generality, we may assume $$I = [k] \makebox{ and } J = \{s-k+1, s-k+2, ..., s\}.$$ 
Let $\lambda \in \F_{q}^*$ 
We take $L_{1,s}(\lambda)\in GL_{\ell}(\F_{q})$. Consider the generic matrix $\Y=L_{1,s}(\lambda)\X L_{s,1}(\lambda)$. We shall prove that the codeword given by $g = f(\Y) = f(L_{1,s}(\lambda)\X L_{s,1}(\lambda))$ has a maximal determinant with a smaller spread. The multilinearity of the determinant implies that $$det_{I,J}(\Y) = det_{I,J}(\X) - \lambda det_{I,J\cup \{1\}-\{s\}}(\X).$$

Note that the only minors of the form $det_{A,B}(\Y)$ such that $det_{I,J\cup \{1\}-\{s\}}(\X)$ may appear are the following:
$$\mathcal{Q} = det_{I\cup \{s\}-\{1\},J\cup \{1\}-\{s\}}(\Y), \mathcal{P} = det_{I\cup \{s\}-\{1\},J}(\Y).$$

Note the spread of $\mathcal{P}$ has size $s-1$, which contradicts that $\mathcal{M}$ is of minimal spread size. This implies $f_{I\cup \{s\}-\{1\},J} = 0$. This implies we only need to worry about $\mathcal{Q}$. However, because $\Y$ is symmetric, $\mathcal{Q} = det_{I,J}(\Y)$. Thus it is not considered as a part of the linear combination. Therefore $\mathcal{N} = det_{I,J\cup \{1\}-\{s\}}(\X) \in supp(g)$ where $\mathcal{N}$ is size k and has spread $s-1 = k$.


\end{proof}
\begin{lemma}\label{lem:spread reduction}
Let $f \in \Fl$. Suppose that $f$ has a maximal minor of minimal spread, $\mathcal{M}$ of size $k$ and spread $=s>k+1$. Then there exists $g \in \Fl$ of the same weight with a maximal minor with size $k$ and spread size $\leq s-1$.
\end{lemma}

\begin{proof}
Let $\mathcal{M} = det_{I,J}(\X)$ be a maximal minor of size $k$ and spread size $s>k$ in the support of $f$.
Applying a suitable permutation of rows and columns, Without loss of generality, we may assume $$I = [k] \makebox{ and } J = \{s-k+1, s-k+2, ..., s\}.$$ 
Note that $s>k+1$ implies $s-1 \in J-I$
Let $\lambda \in \F_{q}^*$ 
We take $L_{1,s-1}(\lambda)\in GL_{\ell}(\F_{q})$. Consider the generic matrix $\Y=L_{1,s-1}(\lambda)\X L_{s-1,1}(\lambda)$. We shall prove that the codeword given by $g = f(\Y) = f(L_{1,s-1}(\lambda)\X L_{s-1,1}(\lambda))$ has a maximal determinant with a smaller support. The multilinearity of the determinant implies that $$det_{I,J}(\Y) = det_{I,J}(\X) - \lambda det_{I,J\cup \{1\}-\{s-1\}}(\X).$$

Note that the only minors of the form $det_{A,B}(\Y)$ such that $det_{I,J\cup \{1\}-\{s-1\}}(\X)$ may appear are the following:
$$\mathcal{Q} = det_{I\cup \{s-1\}-\{1\},J\cup \{1\}-\{s-1\}}(\Y), \mathcal{P} = det_{I\cup \{s-1\}-\{1\},J}(\Y).$$

Note the spread of $\mathcal{P}$ has size $s-1$, which contradicts that $\mathcal{M}$ is of minimal spread size. This implies $f_{I\cup \{s-1\}-\{1\},J} = 0$. This implies we only need to worry about $\mathcal{Q}$ and its corresponding $f_{I\cup \{s-1\}-\{1\},J\cup \{1\}-\{s-1\}}$. By symmetry it is equivalent to the minor given by switching its row and columns set. Note the following:
\begin{itemize}
    \item Rowset of $\mathcal{Q}^T$ is $\{1, s-k+1, s-k+2, ..., s-2, s\}$
    \item Columnset of $\mathcal{Q}^T$ is $\{2, 3, ..., k, s-1\}$
\end{itemize}
Note $1<2$, however $s \nless s-1$. This implies $\mathcal{Q}$ is not a doset minor and thus is not considered to be a part of the linear combination. 
Therefore $\mathcal{N} = det_{I,J\cup \{1\}-\{s-1\}}(\X) \in supp(g)$ where $\mathcal{N}$ is size $k$ and has spread $s-1$.


\end{proof}
As a direct consequence of the previous lemma, we obtain the following:
\begin{corollary}\label{cor:spread s to k}
Let $f \in \Fl$. Suppose that $f$ has a maximal minor of minimal spread, $\mathcal{M}$ of size $k$ and spread $=s>k$. Then we may find $g$ of the same weight with a maximal minor of minimal spread, $\mathcal{N}$ with size $k$ and spread $k$.
\end{corollary}
\begin{proof} We apply Lemma \ref{lem:spread reduction} repeatedly and end with Lemma \ref{lem:spread reduction 1}.
\end{proof}

Note that by Corollary \ref{cor:spread s to k}, we may without loss of generality assume that $f$ has a maximal minor of size $k$ whose spread is $k$. 
Consequently we can stop considering the spread for our bounds. Thus we shall define $w_{\ell,k}$.
\begin{definition}
$w_{\ell,k}$ denotes the minimum weight of a function $f$ such that the matrices are $\ell \times \ell$, with maximal minor of size $k \times k$.
\end{definition}

We finalize this section with the following:
\begin{corollary}\label{cor:weight bound}
Let $f \in \Fl$. Suppose that $f$ has a maximal minor of size $k$ whose spread has size $=s$. Then $wt(f) \geq q^\frac{\ell^2+\ell-k^2-k}{2}(w_{k,k})$.
\end{corollary}

\begin{proof}
Let $f$ have a maximal minor of size $k$ whose spread has size $=s$. By Corollary \ref{cor:spread s to k}, we may apply an automorphism such that we transform $f$ to a function of the same weight with a maximal minor of minimal spread, with size $k$ and spread $k$. By Lemma \ref{lem: reduce to spread}, we may then state that  $wt(f) \geq q^\frac{\ell^2+\ell-k^2-k}{2}(w_{k,k})$. \end{proof}

\section{Calculating $d(C^{\mathbb{S}}(2))$}

For the case $\ell = 2$, recall $\X$ us a generic symmetric matrix of the form 
$$\X = \begin{bmatrix}
    X_{1,1}&X_{1,2}\\X_{1,2}&X_{2,2}
    \end{bmatrix},$$ \\
where $X_{1,1}, X_{1,2}, X_{2,2} \in \F_q$    

In this case a function $f \in \Fl$ is of the form
$$f = f_{\emptyset, \emptyset} + f_{1,1}X_{1,1} + f_{1,2}X_{1,2} + f_{2,2}X_{2,2}+ f_{12,12}(X_{1,1}X_{2,2}-X_{1,2}^2), f_{i,j} \in \F_{q}.$$

We split our proof in two cases: $f_{12,12} = 0$ or $f_{12,12} \neq 0$.
\begin{lemma}
Let $\ell = 2$. Suppose $f \in \Fl$ where $f$ is a nonzero function of the form
$$f = f_{\emptyset, \emptyset} + f_{1,1}X_{1,1} + f_{1,2}X_{1,2} + f_{2,2}X_{2,2}.$$
Then $wt(f) \geq q^3-q^2$

\end{lemma}
\begin{proof}

To determine $wt(f)$ we count the solutions to $$f = f_{\emptyset, \emptyset} + f_{1,1}X_{1,1} + f_{1,2}X_{1,2} + f_{2,2}X_{2,2} = 0.$$ Suppose that $f_{\emptyset, \emptyset}$ is the only nonzero coefficient. this implies $f = f_{\emptyset, \emptyset} \neq 0$. This implies that for all matrices in $\HM$, we have that $f$ is nonzero. This implies $wt(f) = q^3$. Then suppose that some other $f_{i,j}$ is nonzero. Then through the use of code automorphisms we may assume without loss of generality that $f_{1,1} \neq 0$. Then for any of the $q$ values of $X_{1,2}$ and any of the $q$ values of $X_{2,2}$ there is at most $1$ value of $X_{1,1}$ which makes $f = 0.$ Consequently there are at most $q^2$ matrices in $\HM$ such that $f(S) = 0$ and $wt(f) \geq q^3-q^2$. 


\end{proof}

Now consider $f \in \Fl$ where the $2 \times 2$ minor of $\X$ appears in $supp(f)$. Without loss of generality we assume $f_{12,12} = 1$
\begin{lemma}\label{lem:wt22}
Let $\ell = 2$. Suppose $f \in \Fl$ where $f$ is of the form
$$f = f_{\emptyset, \emptyset} + f_{1,1}X_{1,1} + f_{1,2}X_{1,2} + f_{2,2}X_{2,2}+ X_{1,1}X_{2,2}-X_{1,2}^2.$$
Then $wt(f) \geq q^3-q^2-q$

\end{lemma}
\begin{proof}

As in the previous case we count the number of solutions to 
$$f = f_{\emptyset, \emptyset} + f_{1,1}X_{1,1} + f_{1,2}X_{1,2} + f_{2,2}X_{2,2}+ X_{1,1}X_{2,2}-X_{1,2}^2 = 0.$$

We move the terms with $X_{1,1}, X_{2,2}$ to one side and obtain:

$$f_{1,1}X_{1,1} + f_{2,2}X_{2,2} + X_{1,1}X_{2,2} = X_{1,2}^2 -f_{1,2}X_{1,2} -f_{\emptyset, \emptyset}.$$

Now we add $f_{1,1}f_{2,2}$ to both sides:

$$f_{1,1}f_{2,2} + f_{1,1}X_{1,1} + f_{2,2}X_{2,2} + X_{1,1}X_{2,2} = X_{1,2}^2  -f_{1,2}X_{1,2} -f_{\emptyset, \emptyset} +f_{1,1}f_{2,2}.$$

The left side factors as:

$$(X_{1,1}+f_{2,2})(X_{2,2}+f_{1,1}) = X_{1,2}^2 -f_{1,2}X_{1,2} -f_{\emptyset, \emptyset} +f_{1,1}f_{2,2}.$$

The right hand side of the equation is an univariate polynomial in $X_{1,2}$ of degree $2$. Denote by $$P(X_{1,2}) =X_{1,2}^2 -f_{1,2}X_{1,2} -f_{\emptyset, \emptyset} +f_{1,1}f_{2,2}.$$ Let $S = \{ \lambda \in \F_q | P(\lambda) = 0\}$ denote the set of zeroes of $P(X_{1,2})$. Note that $\# S \leq 2$. Let $\alpha \in S.$
In this case $P(\alpha) = 0$. Lemma \ref{lem:hyperboliczeroes} implies that there are $2q-1$ values of $X_{1,1}$ and $X_{2,2}$ such that the equation
$$(X_{1,1}+f_{2,2})(X_{2,2}+f_{1,1}) = P(\alpha)$$ is satisfied. This implies both sides are $0$ for exactly $\#S(2q-1)$ values.

Now assume $\alpha \in \F_q \setminus S.$ In this case Lemma \ref{lem:hyperboliczeroes} implies that there are $q-1$ values of $X_{1,1}$ and $X_{2,2}$ such that the equation
$$(X_{1,1}+f_{2,2})(X_{2,2}+f_{1,1}) = P(\alpha)$$ is satisfied. This implies there are $(q-\# S)(q-1)$ solutions to the equation where $P(\alpha) \neq 0.$

Therefore there are $$\# S (2q-1) + (q-\#S)(q-1)$$ elements of $\HM$ such that $f = 0.$ As $\# S \leq 2$, we have that $f$ has at most $$(2)(2q-1) + (q-2)(q-1) = q^2+q$$
Consequently, $wt(f) \geq q^3-q^2-q$.
\end{proof}

It is very useful to classify the codewords with $2 \times 2$ determinant as a maximal determinant in $supp(f)$.

\begin{lemma}\label{lem:22classfier}
Let $\ell = 2$ and $q$ be odd. Suppose $f \in \Fl$ where $f$ is of the form
$$f = f_{\emptyset, \emptyset} + f_{1,1}X_{1,1} + f_{1,2}X_{1,2} + f_{2,2}X_{2,2}+ X_{1,1}X_{2,2}-X_{1,2}^2.$$ The following statements are true:

If $ {(\frac{1}{2}f_{1,2})}^2-f_{1,1}f_{2,2}+f_{\emptyset, \emptyset} = 0 $ then  $wt(f) \geq q^3 -q^2.$

If $ {(\frac{1}{2}f_{1,2})}^2-f_{1,1}f_{2,2}+f_{\emptyset, \emptyset} \neq 0 $ then  $wt(f) \geq q^3-q^2-q$.

\end{lemma}

\begin{proof}

Suppose $f$ is of the form $$f = f_{\emptyset, \emptyset} + f_{1,1}X_{1,1} + f_{1,2}X_{1,2} + f_{2,2}X_{2,2}+ X_{1,1}X_{2,2}-X_{1,2}^2.$$ As in the proof of Lemma \ref{lem:wt22} we write $f$ as $$f = (X_{1,1} +f_{2,2})(X_{2,2} +f_{1,1}) + P(X_{1,2}) $$ where $$P(X_{1,2}) = -X_{1,2}^2 +f_{1,2}X_{1,2} +f_{\emptyset, \emptyset} -f_{1,1}f_{2,2} $$ is a polynomial of degree $2$.



Now we shall change the variable $X_{1,2}$ to the variable $T_2$ where $X_{1,2} = T_2 + \frac{1}{2}f_{1,2}$. In this case $$ P\left(T_2+ \frac{1}{2}f_{1,2}\right) = -\left(T_2 + \frac{1}{2}f_{1,2}\right)^2 + f_{1,2}^2 +f_{1,2}\left(T_2 + \frac{1}{2}f_{1,2}\right) +f_{\emptyset, \emptyset} -f_{1,1}f_{2,2}.$$

Expanding and eliminating like terms we obtain

$$P\left(T_2 + \frac{1}{2}f_{1,2}\right) = -T_2^2 + {\left(\frac{1}{2}f_{1,2}\right)}^2+f_{\emptyset, \emptyset} - f_{1,1}f_{2,2}. $$
Note that if ${\left(\frac{1}{2}f_{1,2}\right)}^2+f_{\emptyset, \emptyset} - f_{1,1}f_{2,2} = 0$, then $P\left(T_2 +\frac{1}{2}f_{1,2}\right)$ has exactly one zero, namely $T_2 = -\frac{1}{2}f_{1,2}$. 

 Note that the number of zeroes of $P(T_2 + \frac{1}{2}f_{1,2})$ is precisely the same number of solutions to $P(X_{1,2}) = 0$

When $P(X_{1,2}) = 0$ there are $2q-1$ values of $X_{1,1}$ and $X_{2,2}$ which make $f = 0$. When $P(X_{1,2}) \neq 0$, there are $q-1$ values of $X_{1,1}$ and $X_{2,2}$ which make $f = 0$.

Therefore $f = 0$ for at most $2q-1 + (q-1)(q-1) = q^2 -2q +1 +2q -1 =q^2$ and $wt(f) \geq q^3 -q^2.$

If ${(\frac{1}{2}f_{1,2})}^2+f_{\emptyset, \emptyset} - f_{1,1}f_{2,2} \neq 0$, then there are at most $2$ values of $X_{1,2}$ which make $P(X_{1,2}) = 0$. When $P(X_{1,2}) = 0$, there are $2q-1$ values of $X_{1,1}$ and $X_{2,2}$ which make $f = 0$. When $P(X_{1,2})  \neq 0$, there are $q-1$ values of $X_{1,1}$ and $X_{2,2}$ which make the equation true. Therefore $f = 0$ for at most $(2q-1)(2) + (q-2)(q-1) = q^2 +q $ and $wt(f) \geq q^3- q^2 -q.$

\end{proof}

By Berlekamp, Rumsey and Solomon, we have the following:
\begin{proposition}\cite{BERLEKAMP1967553}\label{prop:quadsols}
The quadratic equation $x^2+bx+c = 0$, if $b=0$, has a unique solution in $GF(2^k)$. 
\end{proposition}

\begin{lemma}\label{lem:22classfiereven}
Let $\ell = 2$ and $q$ be even. Suppose $f \in \Fl$ where $f$ is of the form
$$f = f_{\emptyset, \emptyset} + f_{1,1}X_{1,1} + f_{1,2}X_{1,2} + f_{2,2}X_{2,2}+ X_{1,1}X_{2,2}-X_{1,2}^2.$$ The following statements are true:

If $ f_{1,2} = 0 $ then  $wt(f) \geq q^3 -q^2.$

If $ f_{1,2} \neq 0 $ then  $wt(f) \geq q^3-q^2-q$.

\end{lemma}

\begin{proof}

Suppose $f$ is of the form $$f = f_{\emptyset, \emptyset} + f_{1,1}X_{1,1} + f_{1,2}X_{1,2} + f_{2,2}X_{2,2}+ X_{1,1}X_{2,2}-X_{1,2}^2.$$ As in the proof of Lemma \ref{lem:wt22} we write $f$ as $$f = (X_{1,1} +f_{2,2})(X_{2,2} +f_{1,1}) + P(X_{1,2}) $$ where $$P(X_{1,2}) = X_{1,2}^2 +f_{1,2}X_{1,2} +f_{\emptyset, \emptyset} -f_{1,1}f_{2,2} $$ is a polynomial of degree $2$.



By Proposition \ref{prop:quadsols}, if $ f_{1,2} = 0 $, then $P(X_{1,2})$ has exactly one zero. 
When $P(X_{1,2}) = 0$ there are $2q-1$ values of $X_{1,1}$ and $X_{2,2}$ which make $f = 0$. When $P(X_{1,2}) \neq 0$, there are $q-1$ values of $X_{1,1}$ and $X_{2,2}$ which make $f = 0$.

Therefore $f = 0$ for at most $2q-1 + (q-1)(q-1) = q^2 -2q +1 +2q -1 =q^2$ and $wt(f) \geq q^3 -q^2.$

If $ f_{1,2} \neq 0$, then there are at most $2$ values of  $X_{1,2}$ which make $P(X_{1,2}) = 0$.  When $P(X_{1,2}) = 0$, there are $2q-1$ values of $X_{1,1}$ and $X_{2,2}$ which make $f = 0$. When $P(X_{1,2})  \neq 0$, there are $q-1$ values of $X_{1,1}$ and $X_{2,2}$ which make the equation true. Therefore $f = 0$ for at most $(2q-1)(2) + (q-2)(q-1) = q^2 +q $ and $wt(f) \geq q^3- q^2 -q.$

\end{proof}

\section{Calculating $d(C^{\mathbb{S}}(3))$}

The authors in \cite{BGT} specialized from the $\ell \times \ell$ case down to the $(\ell-1) \times (\ell -1 )$ case to find the minimum distance of the affine Grassmann codes. As the matrices in the affine Grassmann code are generic, one can perform a partial evaluation on any row and any column while still preserving the structure of the code $\CA$. In the case of symmetric matrices, symmetry must be preserved. This means that, just as in \cite{gonzalez2021affine}, a partial evaluation on a column also fixes the corresponding row. For $\ell = 3$ we have two cases: spread $< 3$ and spread $=3$

\begin{lemma}\label{lem: 33 small spread weight}
Let $f \in \Fl$. Suppose that $\ell = 3$ and $f$ has a maximal minor whose spread has size $\leq 2$. Then $wt(f) \geq q^6 -q^5-q^4$.
\end{lemma}
\begin{proof}
The lemma follows by applying Corollary \ref{cor:weight bound}
\end{proof}

Now we shall assume $\X$ is of the form

   $$\X = \begin{bmatrix}
    X_{1,1}&X_{1,2}& Y_{1,3}\\X_{1,2}&X_{2,2} & Y_{2,3}\\ Y_{1,3}&Y_{2,3}&Y_{3,3}
    \end{bmatrix},$$ 
    where $X_{1,1}, X_{1,2}, X_{2,2},Y_{1,3}, Y_{2,3}, Y_{3,3}$ are variables of elements in $\F_q$. Now we shall study the $q^3$ possible specializations of $Y_{1,3}, Y_{2,3}, Y_{3,3}$ and their effect on $f \in \Fl$ on the remaining unspecialized $2 \times 2$ submatrix. We remind the reader of the following minor expansion:
    
    \begin{proposition}
$$    det_{123,123}(\X) = det_{12,12}(\X)Y_{3,3} -det_{13,12}(\X)Y_{2,3} + det_{23,12}(\X)Y_{1,3}   $$
    \end{proposition}
    
We will make use of the following known quantity to work out some special cases
    \begin{proposition}\cite{McWilliams}\label{prop:det3}
    The number of symmetric $n \times n$ matrices of full rank with entries in $\F_q$ is equal to
    $$q^{\binom{n+1}{2}}  \prod_{i=1}^{\lceil \frac{n}{2} \rceil} \left(1-\frac{1}{q^{2i-1}}\right) $$
    \end{proposition}
    
As a consequence, we have the following corollary:
\begin{corollary}\label{col:spec det3}
Let $\ell = 3$, if $f = det_{123,123}(\X) +f_{\emptyset, \emptyset}$, then
$$wt(f)  =
\begin{cases}
     q^6-q^5-q^3+q^2 & \makebox{ if } f_{\emptyset, \emptyset} = 0\\
     q^6-q^5+q^2 & \makebox{ if } f_{\emptyset, \emptyset} \neq 0
\end{cases}
$$
\end{corollary}
\begin{proof}
\textbf{Case 1: $f_{\emptyset, \emptyset} = 0$}\\
If $f_{\emptyset, \emptyset} = 0$, then the number of matrices such that $f$ evaluates to a nonzero value is exactly the amount of $3\times 3$ full rank matrices. By Proposition \ref{prop:det3}, $wt(f) = q^6-q^5-q^3+q^2.$

\textbf{Case 2: $f_{\emptyset, \emptyset} \neq 0$}\\
If $f_{\emptyset, \emptyset} \neq 0$, then the number of matrices such that $f$ evaluates to $0$ is exactly the amount of $3\times 3$ full rank matrices whose determinant is $-f_{\emptyset, \emptyset}$. Using Proposition \ref{prop:det3}, this amount is $\frac{q^6-q^5-q^3+q^2}{q-1} = q^5-q^2$. Consequently, the number of matrices which evaluate $f$ to a nonzero value is $q^6-(q^5-q^2).$ Therefore, 

$wt(f) = q^6-q^5+q^2.$
\end{proof}

We shall use symmetric translations to transform the codewords and simplify the calculation of $wt(f)$. Let us denote by $f_{a,b,c}(\X)$ the minor combination obtained by the partial evaluation of $f(\X)$ at $Y_{1,3} = a$, $Y_{2,3}= b$ and $Y_{3,3} = c$.

If q is odd and $f \in \Fl$ such that $det_{123,123}(\X) \in supp(f)$, then
Lemma \ref{lem:minorclear} implies we may assume $f$ has no $2\times 2$ minors. Recall $f_{123,123}\neq 0$ implies we may assume without loss of generality that $f_{123,123} = 1$. Thus $f$ is of the form

$$det_{123,123}(\X) + f_{1,1} X_{1,1} + f_{2,2} X_{2,2} + f_{1,2}X_{1,2}+ f_{3,3} Y_{3,3} + f_{1,3}Y_{1,3}+ f_{2,3}Y_{2,3} + f_{\emptyset, \emptyset}.$$

Note that by Corollary \ref{col:spec det3} if $f = det_{123,123}(\X)$ or $f = det_{123,123}(\X) +f_{\emptyset, \emptyset}$ then we already know the weight of $f$. Thus we assume there's at least one of the following coefficients $f_{1,1}$, $f_{1,2}$, $f_{1,3}$, $f_{2,2}$, $f_{2,3}$ or $f_{3,3}$ which is not zero. We now prove that we may perform automorphisms while preserving this structure of $f.$

\begin{lemma}\label{lem: invertible mats and dets}

Let $f(\X) \in \Fl_t$ and let $A \in GL_\ell(\F_q)$. Then $g(\X) = f(A^T \X A) \in \Fl_t$.
\end{lemma}
\begin{proof}
Let $L_{i,j}(a)$ be the elementary operation of adding $a$  times row $i$ to row $j$. Then $$det_{I,J}(L_{i,j}(a)\X) =    \begin{cases} 
       det_{I,J}(\X) & \makebox { if } j \not\in I \\
       det_{I,J}(\X) +a det_{I\cup \{i\} \setminus\{j\} ,J }(\X) & \makebox { if } j \in I, i \not\in I   \\
              det_{I,J}(\X) & \makebox { if } i,j \in I\\ 
   \end{cases}.$$
   
   Then $$det_{I,J}(\X L_{i,j}(a)^T) =    \begin{cases} 
       det_{I,J}(\X) & \makebox { if } j \not\in J \\
       det_{I,J}(\X) +a det_{I,J\cup \{i\} \setminus\{j\} }(\X) & \makebox { if } j \in I, i \not\in I   \\
              det_{I,J}(\X) & \makebox { if } i,j \in I 
   \end{cases}.$$

Let $T_{i,j}$ be the elementary operation of switching row $i$ and row $j$. Then $$det_{I,J}(T_{i,j}\X) =    \begin{cases} 
       det_{I,J}(\X) & \makebox { if } i,j \not\in I \\
       det_{I\cup \{i\} \setminus\{j\} ,J }(\X) & \makebox { if } j \in I, i \not\in I \\
       det_{I\{j\} \setminus\{i\} ,J  }(\X) & \makebox { if } i \in I, j \not\in I\\
              -det_{I,J}(\X) & \makebox { if } i,j \in I\\ 
   \end{cases}.$$
Similarly
$$det_{I,J}(\X T_{i,j}) =    \begin{cases} 
       det_{I,J}(\X) & \makebox { if } i,j \not\in J \\
       det_{I ,J\cup \{i\} \setminus\{j\} }(\X) & \makebox { if } j \in J, i \not\in J \\
       det_{I ,J\{j\} \setminus\{i\}  }(\X) & \makebox { if } i \in J, j \not\in J\\
              -det_{I,J}(\X) & \makebox { if } i,j \in J\\ 
   \end{cases}.$$

Let $D_{i}(a)$ be the elementary operation of multiplying row $i$ by $a$. Then $$det_{I,J}(D_{i}(a)\X) =    \begin{cases} 
       det_{I,J}(\X) & \makebox { if } i \not\in I \\
       a det_{I,J}(\X)  & \makebox { if } i \in I\\ 
   \end{cases}.$$
   
   Similarly, $$det_{I,J}(\X D_{i}(a)) =    \begin{cases} 
       det_{I,J}(\X) & \makebox { if } i \not\in J \\
       a det_{I,J}(\X)  & \makebox { if } i \in J\\ 
   \end{cases}.$$
Any invertible matrix $A$, may be written as a product of these elementary matrices. This implies that the permutation induced by $A$ is a composition of the previous maps. As all of the maps, take $t$ minors to linear combinations of other $t$ minors, $g(\X) = f(A^T \X A) \in \Fl_t$.

\end{proof}

\begin{corollary}\label{col: full detscalar}
Let $f\in \Fl$ such that $det_{[\ell],[\ell]}(\X) \in supp(f)$ and $\forall det_{I,J}(\X) \in supp(f)$ $ \#I\neq \ell-1$. Let $A \in GL_{\ell}(\F_q)$, and $g(\X) = f(A^T\X A)$. Then  $det_{[\ell],[\ell]}(\X) \in supp(g)$ and $\forall det_{I,J}(\X) \in supp(g)$ $\# I\neq \ell-1$.
\end{corollary}
\begin{proof}
   

Let $f\in \Fl$ such that $det_{[\ell],[\ell]}(\X) \in supp(f)$, $\forall det_{I,J}(\X) \in supp(f)$ $ \#I\neq \ell-1$, $A \in GL_{\ell}(\F_q)$, and $g(\X) = f(A^T\X A)$. Note $det_{[\ell],[\ell]}(\X)$ is mapped to $det_{[\ell],[\ell]}(A^T\X A) = det_{[\ell],[\ell]}(A)^2det_{[\ell],[\ell]}(\X)$. $A \in GL_{\ell}(\F_q)$ implies $det_{[\ell],[\ell]}(A) \neq 0$. Thus, $det_{[\ell],[\ell]}(\X) \in supp(g)$. For all other minors, by Lemma \ref{lem: invertible mats and dets}, we have that $det_{I,J}(\X)$ gets mapped to minors of the same size. Thus, there are no added $\ell-1 \times \ell-1$ minors. Therefore, $det_{[\ell],[\ell]}(\X) \in supp(g)$ and $\forall det_{I,J}(\X) \in supp(g)$ $\# I\neq \ell-1$.

\end{proof}

Corollary \ref{col: full detscalar} implies that we may perform automorphisms of the form $X \rightarrow A^T\X A$ keeping the condition that there are no $2\times 2$ minors in $f$. Thus without loss of generality, we may assume we have $f_{1,1} \neq 0.$

\begin{lemma}\label{lem:33classifier1}
Let $f \in \Fl$, with $q$ odd. Suppose that the maximal minor of $f$ is the full determinant $det_{123,123}(\X)$. If $f_{1,1}\neq 0$ then $$ wt(f) \geq q^6-q^5-q^4+q^3.$$
\end{lemma}
\begin{proof}

Let $f$ be as in the statement of the lemma. We consider what happens when we evaluate $f$ along the third row and column. Now we shall count the number of zeroes of $f$ when specializing $Y_{3,3} = 0.$ The specialization $Y_{1,3} = a, Y_{2,3} = b, Y_{3,3} = 0$ is of the form 

$$f_{a,b,0}(\X) = \begin{bmatrix}
    X_{1,1}&X_{1,2}& a\\X_{1,2}&X_{2,2} & b\\ a&b&0
    \end{bmatrix} +    f_{1,1} X_{1,1} + f_{2,2} X_{2,2} + f_{3,3} (0) + f_{1,2}X_{1,2} +f_{1,3}a+ f_{2,3}b+ f_{\emptyset, \emptyset}.$$

Expanding the $3 \times 3$ determinant we obtain:

$$f_{a,b,0}(\X) =  abX_{1,2} -b^{2}X_{1,1} -a^{2}X_{2,2} + abX_{1,2} + f_{1,1} X_{1,1} + f_{2,2} X_{2,2} + f_{1,2}X_{1,2} + f_{1,3}a+ f_{2,3}b + f_{\emptyset, \emptyset}.$$
Collecting like terms on $X_{1,1}, X_{1,2}, X_{2,2}$ we obtain
$$f_{a,b,0}(\X) =  (f_{1,1}-b^2)X_{1,1} + (f_{1,2}+2ab)X_{1,2} + (f_{2,2}-a^{2})X_{2,2} + f_{1,3}a+ f_{2,3}b + f_{\emptyset, \emptyset}.$$

The resulting polynomial is linear in $X_{1,1}$, $X_{1,2}$ and $X_{2,2}$. Note that the coefficient of the $(1,1)$--minor of the partial specialization $f_{a,b,0}(\X)$ is $f_{1,1} - b^{2}$,  the coefficient of the $(2,2)$--minor is $f_{2,2} - a^{2}$, and the coefficient of the $(1,2)$ minor is $2ab+f_{1,2}$. Lemma \ref{lem:systemsols} implies there are at most $2$ partial specializations such that all three coefficients are $0$. Therefore for the remaining $q^2-2$ specializations we get a nonzero polynomial with at least $q^3-q^2$ nonzeroes.

Suppose now that $f_{a,b,c}(\X)$ is the partial evaluation of $f$ with $Y_{1,3} = a$, $Y_{2,3} = b$ and $Y_{3,3} = c \neq 0.$ Then $f_{a,b,c}(\X)$ expands to:

$$c(X_{1,1}X_{2,2} - X_{1,2}^{2}) + abX_{1,2} -b^{2}X_{1,1} -a^{2}X_{2,2} + abX_{1,2} + f_{1,1} X_{1,1} + f_{2,2} X_{2,2}+ f_{1,2}X_{1,2} + f_{3,3}c + f_{1,3}a+ f_{2,3}b+ f_{\emptyset, \emptyset}.$$

Collecting like terms we obtain:

$$ c(X_{1,1}X_{2,2} - X_{1,2}^{2}) +(f_{1,1} -b^{2})X_{1,1} +(f_{1,2}+ 2ab)X_{1,2} + (f_{2,2}-a^2) X_{2,2}+ f_{1,2}X_{1,2} + f_{3,3}c + f_{1,3}a+ f_{2,3}b+ f_{\emptyset, \emptyset}.$$

Now we shall apply Lemma \ref{lem:22classfier} to determine the weight of each partial evaluation. Lemma \ref{lem:22classfier} implies that if the coefficients of $f_{a,b,c}(\X)$ satisfy $$\left(\frac{1}{2}\frac{2ab+f_{1,2}}{c}\right)^{2} - \frac{f_{1,1}-b^{2}}{c} \frac{f_{2,2}-a^{2}}{c} + \frac{f_{3,3}c + f_{1,3}a+ f_{2,3}b+ f_{\emptyset, \emptyset}}{c} = 0,  $$
    
then the partial evaluation has weight $q^3-q^2$.

and otherwise if $$\left(\frac{1}{2}\frac{2ab+f_{1,2}}{c}\right)^{2} - \frac{f_{1,1}-b^{2}}{c} \frac{f_{2,2}-a^{2}}{c} + \frac{f_{3,3}c + f_{1,3}a+ f_{2,3}b+ f_{\emptyset, \emptyset}}{c} \neq  0$$ then the partial evaluation has weight of at least $q^3-q^2-q$.

Now we count the number of values of $a$ and $b$ such that the partial evaluation has weight either $q^3-q^2-q$ or $q^3-q^2$. After evaluating the parenthesis, and multiplying by $c^2$ we obtain

$$abf_{1,2}+\frac{1}{4}f_{1,2}-f_{1,1}a^2 - f_{2,2}b^2 -f_{1,1}f_{2,2} +acf_{1,3}+bcf_{2,3}+ c^2f_{3,3} + cf_{\emptyset, \emptyset} = 0.$$

Note that for any of the $q-1$ values of $c$ and the $q$ values of $b$, this expression becomes a polynomial over $a$ of degree $2$. This implies there's at most $(q-1)2q = 2q^2-2q$ values such that $f_{a,b,c}(\X)$ has weight $q^3-q^2$. Consequently, there's $(q-1)(q)(q-2) = q^3-3q^2+2q$ values such that $f_{a,b,c}(\X)$ has weight $q^3-q^2-q$. Altogether:

$$wt(f) \geq (q^3-3q^2+2q)(q^3-q^2-q)+(2q^2-2q)(q^3-q^2) + (q^2-2)(q^3-q^2) $$
$$wt(f) \geq q^6-q^5-q^4+q^3 $$
\end{proof}

\begin{lemma}\label{lem:33classifiereven}
Let $f \in \Fl$, with $q$ even. Suppose that the maximal minor of $f$ is the full determinant $det_{123,123}(\X)$. Then $$ wt(f) \geq q^6-q^5-2q^3+3q^2  $$
\end{lemma}
\begin{proof}
Without loss of generality, we may assume $f_{123,123} = 1$. This implies $f$ is of the form:
$$
    det_{123,123}(\X)+ f_{\emptyset, \emptyset} +f_{12,12}(X_{1,1}X_{2,2}+X_{1,2}^2)+ f_{12,13}(X_{1,1}Y_{2,3}+X_{1,2}Y_{1,3})+ f_{12,23}(X_{1,2}Y_{2,3}+X_{2,2}Y_{1,3})
$$
$$
    +f_{13,13}(X_{1,1}Y_{3,3}+Y_{1,3}^2)+
    f_{23,23}(X_{2,2}Y_{3,3}+Y_{2,3}^2)+ f_{1,1} X_{1,1} + f_{2,2} X_{2,2} + f_{1,2}X_{1,2} + f_{3,3} Y_{3,3} + f_{1,3}Y_{1,3}+ f_{2,3}Y_{2,3} .
$$
We shall consider what happens when we evaluate $f$ along the third row and column. That is, we look at specializations from $Y_{1,3} = a$, $Y_{2,3} = b$ and $Y_{3,3} = c$. We shall split this analysis into two cases.\\
\textbf{Case 1: $c \neq f_{12,12}$}\\
    This implies we have a linear combination with the principal $2 \times  2$ determinant $X_{1,1}X_{2,2}+X_{1,2}^2$. In this case the coefficients of $f_{a,b,c}(\X)$ as a polynomial on $X_{1,1}, X_{1,2}, X_{2,2}$ are:
\begin{itemize}
    \item Coefficient of $X_{1,1}$: $f_{1,1} + f_{13,13}c + f_{12,13}b$
    \item Coefficient of $X_{1,2}$: $f_{1,2} + (c+f_{12,12}) + (f_{12,13}+1)a +(f_{12,23}+1)b $
    \item Coefficient of $X_{1,2}^2$: $(c+f_{12,12}) \neq 0$
    \item Coefficient of $X_{2,2}$: $f_{2,2} + f_{23,23}c + f_{12,23}a$
\end{itemize}    
We shall assume $X_{1,1}$ and $X_{2,2}$ take any fixed value and bound the number of zeroes of the resulting quadratic polynomial on $X_{1,2}$. The resulting quadratic polynomial is of the form $\alpha + \beta X_{1,2} +X_{1,2}^2$ where $\beta$ is the coefficient as described above. Note the coefficient of $X_{1,2}$ is at worst linear in $a$ or $b$. Thus there are at most $(q-1)(q)(q-1) = q^3-2q^2+q$ values such that the coefficient is nonzero. By Lemma \ref{lem:22classfiereven}, these evaluations have weight at least $q^3-q^2-q$. Consequently, we have $(q-1)(q)(q) = q^2-q$ values such that the coefficient is $0$. By Lemma \ref{lem:22classfiereven}, these have weight $q^3-q^2$.\\

\textbf{Case 2: $c = f_{12,12}$}\\
In this case, the coefficients look as such:
\begin{itemize}
    \item Coefficient of $X_{1,1}$: $f_{1,1} + f_{13,13}c + f_{12,13}b + b^2$
    \item Coefficient of $X_{1,2}$: $f_{1,2} + (c+f_{12,12})$
    \item Coefficient of $X_{2,2}$: $f_{2,2} + f_{23,23}c + f_{12,23}a + a^2$
\end{itemize}
 At worst, the coefficients of $X_{1,1}$ and $X_{2,2}$ are quadratics with $2$ solutions over $b$ and $a$ respectively. Altogether, there's at most $4$ values such that both coefficients of $X_{1,1}$ and $X_{2,2}$ vanish. Thus there are at least $q^2-4$ values such that $f$ has weight $q^3-q^2$.\\
$wt(f) \geq (q^3-2q^2+q)(q^3-q^2-q)+(q^2-q)(q^3-q^2)+(q^2-4)(q^3-q^2)$\\
$wt(f) \geq q^6-q^5-2q^3+3q^2$

\end{proof}

As a consequence of the previous lemmas, we have the following:
\begin{lemma}\label{lem:33classifier}
Let $f \in \Fl$. Suppose that the maximal minor of $f$ is the full determinant $det_{123,123}(\X)$. Then 
$$ wt(f) \geq q^6-q^5-q^4+q^3  $$
\end{lemma}

\section{Finding $d(\CH)$ for $\ell \geq 4$}

\subsection{Finding $w_{\ell,k}$ with mathematical induction}
Having determined $w_{2,2}$, $w_{3,2}$ and $w_{3,3}$ (the base cases for $2\leq \ell \leq 3$), we shall now calculate $w_{\ell,k}$ for general $\ell$ in a similar manner to \cite{gonzalez2021affine}.

\begin{lemma}\label{lem: induction}
$w_{k,k}\geq q^{\frac{k^2+k}{2}}-q^{\frac{k^2+k}{2}-1}-q^{\frac{k^2+k}{2}-2}+q^{\frac{k^2-k}{2}}-q$
\end{lemma}
\begin{proof}
In Lemma \ref{lem:22classfier} and Lemma \ref{lem:33classifier}, the base cases for $k = 2$ and $k = 3$ are established. We shall now assume the statement of the lemma is true for $2\leq k \leq K$. We shall now prove that the statement holds for $K+1$. Without loss of generality we may assume $f$ has a maximal $K+1 \times K+1$ principal minor in its support.

As we did in the argument for the base cases, we specialize $det(X)$ along the $(K+1)-th$ column. Note there are exactly $(q^{(K+1)}-q^{(K+1)-1})$ values for the $(K+1)-th$ column such that $x_{K+1,K+1} \neq -f_{[K],[K]}$. This leaves us with a non trivial combination in the $K$ case and with an $K \times K$ minor in its support. Therefore, there are $(q^{K+1}-q^{K})$ values for the specialization of the $(K+1)$-th row and column where we specialize into the case $w_{K,K}$.

We shall now consider the specialization where we do not obtain such a $K \times K$ maximal minor. This is the exclusive case where $x_{K+1,K+1} = -f_{[K],[K]}$. 
Hence all remaining $K\times K$ minors specialize into $K-1 \times K-1$ minors.
Now we consider the possibilities for the $K-1 \times K-1$ minors. 

Note that all such minors are of the form  $\mathcal{M}_{[K+1]-\{i,K+1\},[K+1]-\{j,K+1\}}$. Note that if the minor given by ${[K+1]-\{i,K+1\},[K+1]-\{j,K+1\}}$ does not appear in the partial evaluation, then the coefficients of the partial evaluation must satisfy the equation 
$$f_{[K+1]-\{i,K+1\},[K+1]-\{j,K+1\}}+f_{[K+1]-\{i\},[K+1]-\{j\}} = x_{i,K+1} {x_{j,K+1}}.$$
Lemma \ref{lem:systemsols} implies that the system of polynomial equations as stated above has at most $2$ solutions. 
Thus there are at least $q^{(K+1)-1}-2$ values for the partial evaluation on the $K+1$-th column such that we have a non-trivial combination with a $K-1 \times K-1$ minor in its support. These specializations are of weight at least $$q^{K}(w_{K-1,K-1}).$$
We put together all the inequalities and we obtain:

$$w_{K+1, K+1} \geq $$ 
\begin{tabular}{ll}

  & $(q^{K+1}-q^{K})w_{K,K}+ (q^{K}-2)q^{K}w_{K-1,K-1}$ \\
%
    $=$   & $(q^{K+1}-q^{K})w_{K,K}+ (q^{2K}-2q^{K})w_{K-1,K-1}$ \\


 $\geq$ &$(q^{K+1}-q^{K})(q^{\frac{K^2+K}{2}}-q^{\frac{K^2+K}{2}-1}-q^{\frac{K^2+K}{2}-2}+q^{\frac{K^2-K}{2}}-q)+$ \\
  &  $(q^{2K}-2q^{K})(q^{\frac{K^2-K}{2}}-q^{\frac{K^2-K}{2}-1}-q^{\frac{K^2-K}{2}-2}+q^{\frac{K^2-3K+2}{2}}-q)$\\
  $\geq$ &$q^{\frac{K^2+3K+2}{2}}-q^{\frac{K^2+3K+2}{2}-1}-q^{\frac{K^2+3K+2}{2}-2}+2q^{\frac{K^2+K+2}{2}}-2q^{\frac{K^2+K}{2}}+2q^{\frac{K^2+K}{2}-1}+2q^{\frac{K^2+K}{2}-2}-$\\
  & $2q^{\frac{K^2-K+2}{2}}-q^{\frac{K^2-K}{2}}-q^{2K+1}-q^{K+2}+3q^{K+1}$\\


  $\geq$ &$q^{\frac{K^2+3K+2}{2}}-q^{\frac{K^2+3K+2}{2}-1}-q^{\frac{K^2+3K+2}{2}-2}+q^{\frac{K^2-3K}{2}}-q$\\

$\geq$ &$q^{\frac{(K+1)^2+(K+1)}{2}}-q^{\frac{(K+1)^2+(K+1)}{2}-1}-q^{\frac{(K+1)^2+(K+1)}{2}-2}+q^{\frac{(K+1)^2-(K+1)}{2}}-q$\\

\end{tabular}\\
Therefore, by the principle of strong mathematical induction, the bound is met.
\end{proof}

\begin{proposition}\label{prop: mindist}
$w_{k,k}\geq q^{\frac{k^2+k}{2}}-q^{\frac{k^2+k}{2}-1}-q^{\frac{k^2+k}{2}-2}$
\end{proposition}
\begin{proof}
By Lemma \ref{lem: induction}, we have $w_{k,k}\geq q^{\frac{k^2+k}{2}}-q^{\frac{k^2+k}{2}-1}-q^{\frac{k^2+k}{2}-2}+q^{\frac{k^2-k}{2}}-q$.
Note that for $k\geq 2$ we have $ q^{\frac{k^2+k}{2}}-q^{\frac{k^2+k}{2}-1}-q^{\frac{k^2+k}{2}-2}+q^{\frac{k^2-k}{2}}-q \geq q^{\frac{k^2+k}{2}}-q^{\frac{k^2+k}{2}-1}-q^{\frac{k^2+k}{2}-2}$. In fact, the equality is met for $k = 2$.
Therefore:
$$w_{k,k}\geq q^{\frac{k^2+k}{2}}-q^{\frac{k^2+k}{2}-1}-q^{\frac{k^2+k}{2}-2}$$
\end{proof}

Now we are finally ready to prove the main result of this paper
\begin{theorem}
Suppose that $\ell \geq 2$. Then $$d(\CH)  = q^{\frac{\ell^2+\ell}{2}}-q^{\frac{\ell^2+\ell}{2}-1}-q^{\frac{\ell^2+\ell}{2}-2}.$$ 
\end{theorem}
\begin{proof}
Assume $\ell \geq 2$. Let $f \in \Fl$ with a maximal minor of size $k$ and spread $s$. 
By Corollary \ref{cor:weight bound}, $$wt(f) \geq q^\frac{\ell^2+\ell-k^2-k}{2}(w_{k,k}).$$ By Proposition \ref{prop: mindist} $$w_{k,k}\geq q^{\frac{k^2+k}{2}}-q^{\frac{k^2+k}{2}-1}-q^{\frac{k^2+k}{2}-2}.$$ Together  $$wt(f) \geq q^\frac{\ell^2+\ell-k^2-k}{2}(q^{\frac{k^2+k}{2}}-q^{\frac{k^2+k}{2}-1}-q^{\frac{k^2+k}{2}-2}) = q^{\frac{\ell^2+\ell}{2}}-q^{\frac{\ell^2+\ell}{2}-1}-q^{\frac{\ell^2+\ell}{2}-2}.$$ 
This implies $$d(\CH) \geq q^{\frac{\ell^2+\ell}{2}}-q^{\frac{\ell^2+\ell}{2}-1}-q^{\frac{\ell^2+\ell}{2}-2}.$$ Note that for $\ell \geq 2$, let $$f = det_{12,12}(\X)+det_{1,2}(\X).$$ Then $$wt(f) = q^{\frac{\ell^2+\ell}{2}}-q^{\frac{\ell^2+\ell}{2}-1}-q^{\frac{\ell^2+\ell}{2}-2}.$$ Therefore $$d(\CH)  = q^{\frac{\ell^2+\ell}{2}}-q^{\frac{\ell^2+\ell}{2}-1}-q^{\frac{\ell^2+\ell}{2}-2}.$$ 
\end{proof}

Below we show some parameters of the Affine Symmetric Grassmann codes for $\ell =2$ and $\ell = 3.$

\begin{tabular}{cc}
     
    \begin{tabular}{|c|c|c|c|}
        \hline
        q & n & k & $d(C^{\mathbb{S}}(2))$\\
        \hline
        2 & 8 & 5 & 2\\
        \hline
        3 & 27 & 5 & 15\\
        \hline
        4 & 64 & 5 & 95\\
        \hline
        5 & 125 & 5 & 287\\
        \hline
        7 & 343 & 5 & 440\\
        \hline
        8 & 512 & 5 & 639\\
        \hline
        9 & 729 & 5 & 1199\\
        \hline
    \end{tabular}
    & 
    \begin{tabular}{|c|c|c|c|}
        \hline
        q & n & k & $d(C^{\mathbb{S}}(3))$\\
        \hline
        2 & 64 & 14 & 16\\
        \hline
        3 & 729 & 14 & 405\\
        \hline
        4 & 4,096 & 14 & 2,816\\
        \hline
        5 & 15,625 & 14 & 11,875\\
        \hline
        7 & 117,649 & 14 & 98,441\\
        \hline
        8 & 262,144 & 14 & 225,280\\
        \hline
        9 & 531,441 & 14 & 465,831\\
        \hline 
    \end{tabular} \\
\end{tabular}
\\

\section{The dual Code $\CH^\perp$}
In this section we study some properties of the dual code $\CH^\perp$ and its relation to the dual affine Grassmann code $\CA^\perp$. In particular we prove that the minimum distance codewords of the dual codes are similar for both the dual of affine Grassmann codes and dual affine symplectic Grassmann codes. We begin by defining the dual code. We recall the following:
\begin{definition}
Let $C$ be an $[n,k]$ linear code, then its dual code $C^\perp$ is its orthogonal complement as a vector space. That is, $C^\perp$ is an $[n, n-k]$ code such that:

$$C^\perp := \{ h \in \F_q^n \ | \ \sum\limits_{i=1}^n c_ih_i = 0 \ \forall c \in C \}$$  
\end{definition}

Let $\CH^\perp$ denote the Dual of the affine symplectic Grassmann code. Recall that $$\CH^\perp := \{ h \in \F_{q^2}^{\HM} \ | \ \sum\limits_{S\in \HM} h_S f(S) = 0 \ \forall f \in \Fl  \} $$

Then we know the following: $\CH^\perp$ is an $[q^{\frac{\ell^2-\ell}{2}}, q^{\frac{\ell^2-\ell}{2}}-C(\ell+1)]$ code.

In fact we know more about the minimum distance of $\CA^\perp$.
\begin{proposition}\cite[Theorem 17]{BGT2}

Let $\ell \geq 2$. The minimum distance $d(\CA^\perp)$ of the code $\CA^\perp$ satisfies:
$$
d(\CA^\perp)=  \begin{cases} 
      3 & q>2 \\
      4 & q = 2 
   \end{cases}
$$
\end{proposition}

In subsequent work, one of the named authors along with P. Beelen characterized the minimum distance codewords of $\CA^\perp$.

\begin{definition}
Let $f \in \CH^\perp$ $supp(f) = \{S\in \mathbb{S} | c_{S} \neq 0 \}$
\end{definition}

\begin{definition}
Let $k\leq \ell$, we denote $I_k$ as the matrix with an $k\times k$ identity block and the remaining entries are 0. That is $a_{i,i} = 1$ if $1\leq i\leq k$.
\end{definition}

\begin{definition}
We denote $E_{i,j}$ to be the $\ell \times \ell$ matrix which all entries equal 0 except the $(i,j)-th$ entry which equals 1.
\end{definition}

\begin{proposition}
\cite[Theorem 8]{AfGrass}
Let $\ell \geq 2$, let $q > 2$ and let $c\in \CA^\perp$ be a weight 3 codeword with support $supp(c) = \{N_1, N_2, N_3\}$. Then there exists an automorphism such that we may map $c \rightarrow c'$ where $supp(c') = \{ 0, I_1, \alpha I_1\}$ and $\alpha = (\frac{c_{N_2}}{c_{N_1}+c_{N_2}})$.\\
Conversely, given $\alpha \in \F_{q}\setminus \{0,1\}$, there exists a codeword $c\in \CA^\perp$ with $supp(c) = \{ 0, I_1, \alpha I_1\}$. Its nonzero coordinates satisfy
$$ c_{I_1} = \frac{-\alpha}{\alpha -1}c_0  \makebox{ and } c_{\alpha I_1} = \frac{1}{\alpha -1}c_0$$
\end{proposition}

\begin{proposition}
\cite[Theorem 15]{AfGrass}
Let $\ell \geq 2$. Let $q= 2$ and let $c$ be a codeword of $\CA^\perp$ of weight $4$. Suppose that $supp(c) = \{M_1, M_2, M_3, M_4\}$. Then there exists an automorphism such that we may map $c \rightarrow c'$ where $supp(c')$ is one of the following:
\begin{enumerate}
    \item[i]  $\{ 0, E_{1,1}, E_{1,2}, I_1+E_{1,2}\}$
    \item[ii] $\{ 0, E_{1,1}, E_{2,1}, I_1+E_{2,1}\}$
    \item[iii] $\{ 0, E_{1,1}, E_{1,2}+E_{2,1}, E_{1,1}+E_{1,2}+E_{2,1}\}$
\end{enumerate}

\end{proposition}

\begin{definition}
Let $C$ be a linear code and $T$ be a set of coordinates in $C$. We define the puncturing of $C$ on $T$ as the resulting linear code $C^T$ from deleting all coordinates in $T$ in each codeword of $C$.
\end{definition}
\begin{definition}
Let $C$ be a linear code, $T$ be a set of coordinates in $C$ and $C(T)$ the set of codewords which are 0 on $T$. We define the shortening of $C$ as the puncturing of $C(T)$ on $T$.
\end{definition}
We remark that these code operations are duals of each other. That is, the dual code of  puncturing $C$ is shortening $C^\perp$.

With the fact that $\CH$ is a puncturing of $\CA$ on the positions outside of $\HM$, and that $\CH^\perp$ is the linear code obtained by shortening $\CA^\perp$ at the positions outside of $\HM$ we shall determine the minimum distance of $\CH^\perp$. 


\begin{theorem}
Let $\ell \geq 2$. The minimum distance $d(\CH^\perp)$ of the code $\CH^\perp$ satisfies:
$$
d(\CH^\perp)=  \begin{cases} 
      3 & q>2 \\
      4 & q = 2 
   \end{cases}
$$
\end{theorem}
\begin{proof}
Recall we are puncturing the code $\CA$ at the matrices in $\mathbb{M}^{\ell \times \ell}(\F_{q^2}) \setminus \HM$ to obtain $\CH$. This implies we are shortening the code $\CA^\perp$ to obtain $\CH^\perp$. By \cite{AfGrass}, this implies we have a lower bound$$
d(\CH^\perp)\geq  \begin{cases} 
      3 & q>2 \\
      4 & q = 2 
   \end{cases}.
$$\\

If $q>2$, there exists $c\in \CA^\perp$ such that $supp(c) = \{ 0, E_{1,1}, \alpha E_{1,1} \}$ and $\alpha \in \F_q$. Because all 3 matrices are symmetric, this implies that when shortening we have a codeword $c' \in \CH^\perp$ such that $supp(c) = supp(c')$. Therefore, for $q>2$, $d(\CH^\perp)= 3$.

If $q=2$, there exists $c\in \CA^\perp$ such that $supp(c) = \{ 0, E_{1,1}, E_{1,2}+E_{2,1}, E_{1,1}+E_{1,2}+E_{2,1}\}$. Because all 4 matrices are symmetric, this implies that when shortening we have a codeword $c' \in \CH^\perp$ such that $supp(c) = supp(c')$. Therefore, for $q=2$, $d(\CH^\perp)= 4$.

\end{proof}
As in the case for affine Grassmann codes, we characterize all minimum distance codewords of $\CH^\perp$. Both cases for $q > 2$ and $q = 2$ can be done simply by considering the shortening operation. As there are minimum distance codewords of $\CA^\perp$ over $q$ whose support is entirely of symmetric matrices. It is these codewords which are the minimum distance codewords of $\CH^\perp$.

\section{Conclusion}

In this manuscript we have introduced the affine symplectic Grassmann codes. These are linear codes associated to the affine part of the polar symplectic Grassmannian, defined in the same way affine Grassmann codes are defined from the Grassmannian. As might be expected, the affine symplectic Grassmann code is very similar to the affine Grassmann code and the affine Hermitian Grassmann code \cite{gonzalez2021affine}, except it is defined over symmetric matrices. In fact, the minimum weight codewords of the affine symplectic Grassmann code are similar to the minimum weight codewords of the affine Hermitian Grassmann code. Moreover, all three have similar minimum weight codewords in their dual codes
and their automorphism groups are very similar.

\section*{Acknowledgments}
The authors are very thankful to Sudhir Ghorpade for his insightful comments on the automorphism group and his suggestions to improve this manuscript. This research is supported by NSF-DMS REU 1852171: REU Site: Combinatorics, Probability, and Algebraic Coding Theory and NSF-HRD 2008186: Louis Stokes STEM Pathways and Research Alliance: Puerto Rico-LSAMP - Expanding Opportunities for Underrepresented College Students (2020-2025)

\nocite{Sticht}
\nocite{BGT}
\nocite{AfGrass}
\nocite{gonzalez2021affine}

\bibliographystyle{ieeetr}
\bibliography{references}

\begin{thebibliography}{10}

\bibitem{Nogin}
D.~Y. Nogin, {\em Codes associated to Grassmannians in: "Arithmetic, Geometry,
  and Coding Theory"}, pp.~145--154.
\newblock (Luminy, 1993) Walter De Gruyter, Berlin/New York, 1996.

\bibitem{BGT}
P.~{Beelen}, S.~R. {Ghorpade}, and T.~{H\o{}holdt}, ``Affine {G}rassmann
  codes,'' {\em IEEE Transactions on Information Theory}, vol.~56, no.~7,
  pp.~3166--3176, 2010.

\bibitem{BGT2}
P.~Beelen, S.~R. Ghorpade, and T.~H\o{}holdt, ``Duals of affine {G}rassmann
  codes and their relatives,'' {\em IEEE Transactions on Information Theory},
  vol.~58, no.~6, pp.~3843--3855, 2012.

\bibitem{AfGrass}
P.~Beelen and F.~Piñero, ``The structure of dual {G}rassmann codes,'' {\em
  Designs, Codes and Cryptography}, vol.~79, no.~3, pp.~451--470, 2016.

\bibitem{Cardinali_2013}
I.~Cardinali and L.~Giuzzi, ``Codes and caps from orthogonal {G}rassmannians,''
  {\em Finite Fields and Their Applications}, vol.~24, p.~148–169, Nov 2013.

\bibitem{Symplectic}
I.~Cardinali and L.~Giuzzi, ``Minimum distance of symplectic {G}rassmann
  codes,'' {\em Linear Algebra and its Applications}, vol.~488, 03 2015.

\bibitem{CARDINALI20161924}
I.~Cardinali, L.~Giuzzi, K.~V. Kaipa, and A.~Pasini, ``Line polar grassmann
  codes of orthogonal type,'' {\em Journal of Pure and Applied Algebra},
  vol.~220, no.~5, pp.~1924--1934, 2016.

\bibitem{Cardinali_2018}
I.~Cardinali and L.~Giuzzi, ``Minimum distance of orthogonal line-{G}rassmann
  codes in even characteristic,'' {\em Journal of Pure and Applied Algebra},
  vol.~222, p.~2975–2988, Oct 2018.

\bibitem{Cardinali}
I.~Cardinali and L.~Giuzzi, ``Line {H}ermitian {G}rassmann codes and their
  parameters,'' {\em Finite Fields and Their Applications}, vol.~51,
  p.~407–432, May 2018.

\bibitem{Delsarte1970OnGR}
P.~Delsarte, J.-M. Goethals, and F.~J. MacWilliams, ``On generalized
  reed-muller codes and their relatives,'' {\em Inf. Control.}, vol.~16,
  pp.~403--442, 1970.

\bibitem{Shafiei2013ApolarityFD}
S.~M. Shafiei, ``{Apolarity for determinants and permanents of generic
  matrices},'' {\em Journal of Commutative Algebra}, vol.~7, no.~1, pp.~89 --
  123, 2015.

\bibitem{CLO2007}
D.~Cox, J.~Little, and D.~O’Shea, ``Ideals, varieties, and algorithms. an
  introduction to computational algebraic geometry and commutative algebra,''
  2007.

\bibitem{BERLEKAMP1967553}
E.~Berlekamp, H.~Rumsey, and G.~Solomon, ``On the solution of algebraic
  equations over finite fields,'' {\em Information and Control}, vol.~10,
  no.~6, pp.~553--564, 1967.

\bibitem{gonzalez2021affine}
F.~Pi{\~{n}}ero{-}Gonz{\'{a}}lez and D.~Rivera{-}Laboy, ``Affine hermitian
  grassmann codes,'' {\em CoRR}, vol.~abs/2110.08964, 2021.

\bibitem{McWilliams}
F.~J. MacWilliams and N.~J.~A. Sloane, {\em The Theory of Error-Correcting
  Codes}.
\newblock North-Holland Mathematical Library 16, North-Holland, 1st~ed., 1977.

\bibitem{Sticht}
H.~Stichtenoth, ``On the dimension of subfield subcodes,'' {\em IEEE
  Transactions on Information Theory}, vol.~36, no.~1, pp.~90--93, 1990.

\end{thebibliography}



\end{document}